\documentclass[12pt]{article}

\usepackage{amsthm,amsmath,amsfonts,algorithm,graphicx, braket, cite, mathtools, pifont, caption}
\usepackage{fullpage,authblk}
\usepackage{xcolor}
\usepackage{hyperref}
\newtheorem {theorem} {Theorem}
\newtheorem {corollary} {Corollary}[section]
\newtheorem {lemma} {Lemma}

\theoremstyle{definition}
\newtheorem{proposition}{Proposition}

\theoremstyle{plain}
\usepackage[utf8x]{inputenc}
\usepackage{float}

\newcommand{\ka}{\ket{a}}

\newcommand{\kphibe}{\ket{\phi}_{BE}}

\newcommand{\mq}{ \sqrt{q(1 - q)}}
\newcommand{\sqd}{\sqrt{D}}

\newcommand{\kb}[2]{\ket{#1}\bra{#2}}
\newcommand{\bk}[2]{ \braket{#1 | #2}}

\newcommand{\expmt}{\textbf{\texttt{Exp}}}

\newcommand{\half}{\frac{1}{2}}
\newcommand{\ot}{\otimes}

\newcommand{\mbi}{\mathbb{I}}

\newcommand{\eba}{e_b^a}
\newcommand{\ebap}{e_b^{a'}}

\newcommand{\ebpap}{e_{b'}^{a'}}

\newcommand{\sqobd}{\frac{1}{\sqrt{D}}}
\newcommand{\obd}{\frac{1}{D}}
\newcommand{\ab}{\alpha \beta}

\DeclareMathOperator{\Tr}{Tr}

\newcommand{\ezz}{e_0^0}
\newcommand{\ezo}{e_0^1}

\newcommand{\eoz}{e_1^0}
\newcommand{\eoo}{e_1^1}

\newcommand{\onebyd}{\frac{1}{D}}

\newcommand{\qbtwo}{ \frac{q}{2}}
\newcommand{\omqbtwo}{ \frac{1 - q}{2}}

\newcommand{\vgm}{\vec{\gamma}}

\newcommand{\vmin}{v_{\gmo}^{min}}
\newcommand{\vmax}{v_{\gmo}^{max}}

\newcommand{\vgmmin}{\vgm_1^{min}}
\newcommand{\vgmmax}{\vgm_1^{max}}

\newcommand{\gmo}{\gamma_1}

\DeclarePairedDelimiter\autobracket{(}{)}
\newcommand{\br}[1]{\autobracket*{#1}}

\newcommand{\hbzesig}{H(B^Z | E)_\sigma}
\newcommand{\hbzerho}{H(B^Z | E)_\rho}

\newcommand{\Z}{\mathcal{Z}}
\newcommand{\X}{\mathcal{X}}

\newcommand{\modefull}{\texttt{MODE$=$FULL}}
\newcommand{\modepartial}{\texttt{MODE$=$PARTIAL}}

\newcommand{\zbasis}{\{\ket{0}, \ket{1}, ..., \ket{D - 1}\}}
\providecommand{\keywords}[1]{\textbf{\textit{Index terms---}} #1}

\floatname{algorithm}{Protocol}

\title{New Security Proof of a Restricted High-Dimensional QKD Protocol}
\author[1]{Hasan Iqbal \thanks{hasan.iqbal@uconn.edu}}
\author[1]{Walter O. Krawec}
\affil[1]{Department of Computer Science and Engineering,  University of Connecticut, Storrs, CT 06269, USA}

\begin{document}
	\maketitle
	
	\begin{abstract}
		High-dimensional states (HD) are promising for quantum key distribution (QKD) due to their noise tolerance and efficiency. However, creating, and measuring, HD states is technologically challenging, thus making it important to study HD-QKD protocols where Alice and Bob are restricted in their quantum capabilities. In this paper, we revisit a particular HD-QKD protocol, introduced in (PRA 97 (4):042347, 2018), which does not require Alice and Bob to be capable of sending and measuring in full mutually unbiased bases. In a way, the protocol is a HD version of the three state BB84: one full basis is used for key distillation, but only a single state is used, from an alternative basis, for testing the fidelity of the channel.

          The previous proof of security for this protocol has relied on numerical methods, making it difficult to evaluate for high dimensions.  In this work, we provide a new proof of security and  evaluate the key-rate for high dimensional states, beyond what could previously be evaluated.  Furthermore, our new proof produces better results than prior work for dimensions greater than eight, and shows that HD-states can benefit restricted protocols of this nature.
	\end{abstract}

	
	\section{Introduction}
	Quantum key distribution(QKD) offers provable unconditional security of a shared secret key between two parties Alice and Bob, against an adversary Eve, whose capabilities are only limited by the laws of nature. QKD is one of the more mature areas of quantum information theory \cite{nielsen2002quantum, watrous2008theory, wilde2013quantum}  and has been studied extensively over the last few decades; for a survey, the reader is referred to \cite{qkd-survey-pirandola,qkd-survey-scarani,qkd-survey-omar}. Along with the importance of studying the limits of possibility in the quantum realm, QKD protocols also have crucial real-life significance. Indeed, common cryptographic systems with computational assumptions in use today, may be vulnerable to computational speedups in emerging quantum algorithms \cite{jordan}, coupled with increasingly advancing quantum technologies such as quantum computers \cite{arute2019quantum, chow2021ibm}. To ensure the continuance of security and privacy of our data against this backdrop, designing and analyzing the security of QKD protocols are critical, as these protocols are going to be an integral part of our upcoming quantum-proof communication networks \cite{poppe2008outline, liao2018satellite, awschalom2020long, ribezzo2022deploying}.

	However, implementing QKD systems is hard, due to the inherent difficulties of working with quantum resources \cite{brassard2000security, diamanti2016practical}. So it is important for the protocol designers to be prudent in how much quantum resources they use in a protocol. With this consideration in mind, people have tried to build such protocols with fewer quantum resources\cite{bennett1992quantum, bbm92, boyer2007quantum, zou2009semiquantum}.  Along this line is the so-called ``3-state-BB84'' protocol \cite{mor1998no,fung2006security,branciard2006zero} where the difference with the standard BB84 protocol is that Alice can only send one $\X$ basis state. This protocol and result was extended by Tamaki et al. in \cite{tamaki2014loss} where they examined the case when Alice's state preparation source is extremely noisy and she ends up preparing imperfect states to send to Bob and only sends one of the $\X$ basis states.


	Nurul et. al., \cite{islam2018securing} extended Tamaki et al.'s \cite{tamaki2014loss} result for higher dimensional states where the dimension of the Hilbert spaces that Alice and Bob use, may be greater than two. We denote this variant of the 3-State-BB84 protocol as ``HD-3-State-BB84'' protocol (though it is technically a $D+1$ state protocol, it is inspired by the three-state BB84 protocol, thus the name choice here). As opposed to the two-dimensional or qubit based cases, high dimensional or qudit based QKD systems \cite{hd-survey} exhibit higher noise tolerance in security analysis \cite{bechmann2000quantum, cerf2002security, acin2003security,chau2005unconditionally, sheridan2010security, chau2015quantum,vlachou2018quantum, iqbal2021analysis,yao2022quantum }  and there are remarkable advancements in actual implementations recently \cite{mower2013high, islam2017provably,lee2019large, wang2020high, vagniluca2020efficient,da2021path}. Of course, proving the security of a QKD protocol against the most general form of attack, otherwise known as the coherent attacks, remains one of the challenging aspects of any QKD system. The authors in \cite{islam2018securing} have used a numerical optimization-based method, more precisely, semi-definite programming \cite{vandenberghe1996semidefinite}, to upper bound the phase-error rate of the HD-3-State-BB84, which is necessary to prove security using the Shor-Preskill method \cite{shor2000simple}. People have used SDP-based analysis to analyze QKD systems before  \cite{bunandar2020numerical, winick2018reliable, moroder2006one, george2021numerical} in different contexts. However, due to the prohibitive computational complexity of SDP-based numerical optimization in this case for $D > 10$, analyzing noise tolerance in this method \cite{islam2018securing} becomes less compelling. So, one may look for analytical results for arbitrary dimensions that avoid computational complexities.
	
	In this work, we take the same higher-dimensional variant of the three-state protocol (HD-3-State-BB84) considered by Nurul et. al. \cite{islam2018securing} and analyze its security. Our proof is general enough to work under any noise model - for depolarizing channels we are able to derive an explicit analytical equation while for other channels we reduce the problem to finding eigenvalues of a relatively simple matrix which can be done quickly numerically.  Our new analysis method removes the need for computational optimization, and provides a fast way to compute achievable key rate in arbitrary dimensions, depending only on the dimensions of the Hilbert spaces $D$ and channel noise parameter $q$.
	
	We make several contributions in this work.  First, we revisit a high-dimensional protocol, introduced in \cite{islam2018securing}, and derive a new, information theoretic, security analysis for it. Unlike prior work in \cite{islam2018securing}, our method is analytical and can be used to analyze the key-rate of the protocol for any dimension (though we only derive an exact equation for the case of a depolarization channel noise model - the most often analyzed case in QKD security - other channels can be analyzed easily with our method also).  Previous work, relying on numerical optimization methods, could only realistically be evaluated for relatively small dimension.  The protocol we analyze is a HD-variant of the three-state BB84 protocol and only requires Alice to send a single $\X$ basis state.  Furthermore, we also consider a simpler version of the protocol where Bob only needs to be able to distinguish a single high-dimensional superposition state, making his measurement apparatus significantly simpler.  This would be important should an experimental implementation of this protocol be constructed ever.  Indeed, oftentimes, being able to distinguish all $D$ states of one basis is easy, while being able to also distinguish all $D$ states from an alternative basis is technologically challenging (as with, for instance, time-bin encoding where performing a $Z$ basis measurement requires only recording the time of arrival, while making a measurement in an alternative, mutually unbiased basis, requires cascading interferometers \cite{brougham2013security}).
	
	Our proof makes use of several techniques, based on extensions of methods we developed in \cite{iqbal2020high} for an alternative HD-QKD protocol, that may be broadly applicable to other (HD)-QKD protocols that lack certain ``symmetries'' which often make security proofs easier.  We first design a ``toy'' protocol and analyze its security; we then take advantage of the continuity of von Neumann entropy to promote this analysis to the actual QKD protocol.  Along the way, we prove a new continuity bound for conditional von Neumann entropy of certain quantum states which is tighter than previous work.  While our bound is only applicable to certain forms of classical-quantum systems, our method of proof for this continuity bound may be useful to prove more general results in the future.
	
	Finally, we evaluate our resulting key-rate bounds for this protocol and compare, when possible, to prior work by \cite{islam2018securing}. Importantly, unlike prior work, which showed a decreasing trend in the noise tolerance of the protocol as the dimension increased, we actually show an increasing trend in noise tolerance as the dimension increases.  Thus, in this work, we prove that high-dimensional states can indeed benefit this restricted three-state style protocol - \emph{a previously open question}.  Since high-dimensional states have been shown to benefit the theoretical performance of other QKD protocols \cite{hd-survey}, our work here sheds further light on the potential benefits of HD states to ``simpler'' protocols such as the one analyzed in this paper.

	\section{Preliminaries}
	A density operator $\rho$ describing a quantum state is a Hermitian, positive semi-definite operator with unit trace. In this work, subscripts of a density operator $\rho$, such as in $\rho_{ABE}$, means that it is a quantum system shared among parties Alice, Bob, and Eve. We may drop these subscripts if the context is clear. We use $H(A)_\rho$ to denote the von Neumann entropy of $\rho_A$. The trace norm or the Schatten $p$-norm for $p = 1$ of a density operator $\rho$ is defined as $\|\rho\|_1 := \Tr(\sqrt{\rho \rho^\dagger})$. The trace distance between two density operators $\rho$ and $\sigma$ is defined as $\half \|\rho - \sigma\|_1$.
	We denote the binary entropy function as $h(x) = -x \log x - ( 1- x) \log (1 - x)$ for $0 \le x \le 1$, where the logarithm is of base 2, which is also true for all logarithms used in this work. The Shannon entropy of a probability distribution $\vec{p}$ with $n$ outcomes, is denoted in the usual way as: $H(\vec{p}) = -\sum_{i = 1}^n p_i \log p_i$. The conditional entropy of a bipartite quantum system $\rho_{BE}$, denoted by $H(B |E)_\rho$ is defined as $H(B | E)_\rho = H(BE)_\rho - H(E)_\rho$, where again, we may drop the subscript. We use notation $\Z = \zbasis$ to denote the $D$-dimensional computational basis and use $\X$ to denote the Fourier basis; i.e., $\X = \{ \mathcal{F} \ket{0}, \mathcal{F} \ket{1}, ..., \mathcal{F} \ket{D - 1}\} = \{\ket{x_0},\cdots, \ket{x_{D-1}}\}$. The definition of Fourier transform operator $\mathcal{F}$ is the following:
	\begin{align*}
		\ket{x_a} = \mathcal{F} \ka  = \frac{1}{\sqd} \sum_{b = 0}^{D - 1} \exp\left(\frac{-\pi i a b}{D} \right) \ket{b}.
	\end{align*}
	When a qudit is sent from Alice to Bob, adversary Eve may attack the qudits in different strategies \cite{qkd-survey-pirandola, qkd-survey-scarani}. In one such strategy, namely collective attacks, Eve acts on the traveling qudits in an i.i.d manner, but she may measure her own memory register coherently and at any future time, potentially even after the classical communication part between Alice and Bob is complete. This class of attacks is widely considered in the QKD literature \cite{biham2002security, heid2006efficiency, acin2007device, branciard2008upper, scarani2008quantum, pirandola2014quantum} and has been proven to imply security against general attacks in certain protocols \cite{renner2005information, christandl2009postselection}. Under such an attack scenario, if we have a tripartite density operator $\rho_{A^ZB^ZE}$ where random variables $A^Z, B^Z$ are the classical results of measuring quantum memories $A$ and $B$ in the $\Z$ basis, it was shown by Devetak and Winter \cite{devetak2005distillation, renner2005information, qkd-survey-pirandola} that the key rate $K$ in an asymptotic scenario, is the following:
	\begin{align}
		K = \inf[H(B^Z|E) - H(B^Z|A^Z)],  \label{eq:winterkeyrate}
	\end{align}
	where the infimum is taken over all possible collective attacks that could be performed by Eve which agree with the observed channel statistics.

        An important concept in quantum information theory, and especially useful in quantum cryptography, are \emph{entropic uncertainty relations} (see \cite{coles2017entropic} for a survey) which relate the amount of uncertainty in a system based on the measurement overlap of the measurements performed.  In particular, let $\rho_{AE}$ be a quantum state where the $A$ register is $d$-dimensional.  Then, it was shown in \cite{berta2010uncertainty} that if a measurement is performed on the $A$ register in either the $Z$ (resulting in random variable $A_Z$) or $X$ (resulting in random variable $A_X$) basis, it holds that:
        \begin{equation}\label{eq:eu-1}
          H(A_Z|E) + H(A_X) \ge \log_2D
        \end{equation}
        Note that, actually, the result in \cite{berta2010uncertainty} is more general, however we only present the above version as that is the one which will be important to our work later.

        In \cite{krawec2020new}, an alternative entropic uncertainty relation was shown where, given $\rho_{AE}$, if a measurement in either the $Z$ basis, or a restricted POVM measurement of the form $\{\kb{x_0}{x_0}, I-\kb{x_0}{x_0}\}$ is performed, then it holds that:
        \begin{equation}\label{eq:eu-2}
          H(A_Z|E) + \frac{H_D(Q_X)}{\log_D2} \ge \log_2D,
        \end{equation}
        where $H_D(x)$ is the $D$-ary entropy function:
        \begin{equation}
          H_D(x) = x\log_D(D-1) - x\log_D(x) - (1-x)\log_D(1-x)
        \end{equation}
        and $Q_X$ is the probability of receiving outcome $I-\kb{x_0}{x_0}$.  Note that Equation \ref{eq:eu-2} follows immediately from Theorem 2 in \cite{krawec2020new} by the asymptotic equipartition property of quantum min entropy \cite{tomamichel2009fully}.  See Corollary \ref{cor:eu-new}.
	
	\section{The Protocol}
	
	The protocol is a high-dimensional variant of the three-state BB84 protocol, introduced in \cite{islam2018securing}.  While that paper considered a larger class of protocol, here we consider the ``simplest'' version.  We also consider two versions of the protocol: $\modefull$ and $\modepartial$.  There are two $D$-dimensional bases $\Z$ and $\X$ as defined previously and one distinguished $\X$ basis state we denote simply $\ket{x_0}$ (though it may be any of the $\X$ basis states so long as the choice is public knowledge).  Alice is able to send any $\Z$ basis state or the distinguished $\ket{x_0}$ state; Bob is able to perform (1) a full $\Z$ basis measurement or (2) he can measure in the full $\X$ basis (if $\modefull$) or he can only distinguish $\ket{x_0}$ from any of the other $\ket{x_i}$, $i > 0$ (if $\modepartial$).  That is, he is always able to measure in the full $\Z$ basis, but he does not need to be able to perform a full $\X$ basis measurement if $\modepartial$.
	
	The prepare-and-measure version of the protocol has Alice choosing, independently at random on each round, whether that round is a Key Round or a Test Round.  In the former case, she chooses one of the $\ket{a}\in\Z$ states to prepare and send.  Otherwise, if this is a Test Round, she sends $\ket{x_0}$.  When Bob receives a quantum state, he independently chooses to measure in the $\Z$ basis or in the $\X$ basis (though, as stated, in the latter case he need only be able to distinguish $\ket{x_0}$ from the other $\X$ basis states if $\modepartial$).  Alice and Bob disclose over the authenticated channel what basis choice they each made.  If Alice and Bob both used the $\Z$ basis, they share a key digit; otherwise, they use the round for testing the fidelity of the channel (Bob should always receive $\ket{x_0}$, and any other outcome will represent noise).  Later, for a randomly chosen subset of Key Rounds, Alice and Bob will disclose their full choices and measurement outcomes to determine the $\Z$ basis error rate of the channel.  Should the noise in either basis be ``too large'' (to be discussed), the parties abort.
	
	In our security proof, we actually analyze a partial-entanglement-based version of the protocol.  Here, Alice prepares either an entangled state, for Key Rounds, or a separable one, for Test Rounds.  The protocol is formally stated in Protocol \eqref{prot:hd-3-State-BB84}.  It is trivial to see that the security of this partial entanglement-based protocol will imply the security of the prepare and measure one.

	\begin{algorithm}
		\caption{High-dimensional Restricted BB84 (HD-3-State-BB84) \label{prot:hd-3-State-BB84}}
		$ $\newline
		\textbf{Public Parameters:} The dimension of the Hilbert space $D \ge 2$ and the bases $\Z$ and $\X$ as well as Alice's choice of the single $\X$ basis state $\ket{x_0} = \sqobd \sum_{a = 0}^{D -
			1} \ka$. Also the protocol \texttt{MODE}, namely $\modefull$ or $\modepartial$.
		$ $\newline
		\textbf{Quantum Communication Stage:} The quantum communication stage of the protocol will repeat the following until a sufficiently large classical raw-key is obtained by Alice and Bob:
		\begin{enumerate}
			\item Alice chooses randomly whether a particular round is a key generating round (Key Round) or a channel estimation round (Test Round). If it is a Key Round, she prepares an entangled state $\ket{\psi}_{AT} \coloneqq \frac{1}{\sqrt{D}} \sum_{a = 0}^{D - 1} \ket{a, a}_{AT}$ and sends the $T$ register to Bob through the communication channel. Otherwise, in a Test Round, she sends $\ket{x_0}$.
			
			\item Bob chooses randomly to measure in the $\Z$ basis or the $\X$ basis.  If the latter, and if $\modefull$, he measures in the full $X$ basis; otherwise, if $\modepartial$, he actually measures using the two outcome POVM $\{\ket{x_0}\bra{x_0}, I - \ket{x_0}\bra{x_0}\}$.
			
			\item Alice informs Bob over the authenticated channel about the choice for this round, whether it is a Key Round or a Test Round. Bob tells Alice his choice of measurement basis. If this is a Key Round and Bob uses the $\Z$ basis, then Alice measures her own register also in the $\Z$ basis, in which case, this round can contribute towards the raw key. Otherwise, if this is a Test Round and Bob uses the $\X$ basis, then this round can contribute towards estimating Eve's disturbance.
		\end{enumerate}
		\textbf{Classical Communication Stage: } Alice and Bob proceed with error correction and privacy amplification to obtain a secret key if the protocol was not aborted.
	\end{algorithm}

	
	\section{Security Analysis}
	The ultimate goal of our security analysis is to obtain a lower bound on the achievable key rate for protocol \eqref{prot:hd-3-State-BB84} according to equation \eqref{eq:winterkeyrate}. However, in equation \eqref{eq:winterkeyrate}, the entropy involving Eve's quantum memory $E$ and Bob's classical random variable $B^Z$, denoted as $H(B^Z | E)$, is not straightforward to calculate. Our goal in this section is to calculate a lower bound on this quantity $H(B^Z | E)$, resulting from protocol \eqref{prot:hd-3-State-BB84}. Note that the other quantity $H(B^Z|A^Z)$ in equation \eqref{eq:winterkeyrate} is directly observable by Alice and Bob, as this only involves their $\Z$ basis measurement results, and with classical communications, they can estimate this quantity. We proceed with our security analysis of protocol \eqref{prot:hd-3-State-BB84} under a collective attack scenario according to which the key rate equation\eqref{eq:winterkeyrate} was derived. We only consider ideal single qubit, lossless channels in this work and leave defense against practical attacks \cite{huang2018implementation,diamanti2016practical, lo2014secure} as interesting future works.
	
	The method that we are using to analyze the security of protocol \eqref{prot:hd-3-State-BB84} is based on the works in \cite{krawec2018key, iqbal2020high}. We prove the security of protocol \eqref{prot:hd-3-State-BB84} in three steps.  First, we consider an alternative ``toy'' protocol which is actually a form of quantum random number generation protocol, where Bob gets a random string, but Alice does not get any information.  We work out the resulting quantum states for the toy protocol and the real one, assuming Eve performed the same attack against both.  Next, we analyze the von Neumann entropy of the toy protocol; for this, we can use entropic uncertainty \cite{berta2010uncertainty}.  Finally, we take advantage of the continuity of von Neumann entropy \cite{winter2016tight} to promote the analysis from the toy protocol to the real one.  Along the way, we also derive a new continuity bound for von Neumann entropy for certain types of quantum states.

	\subsection{First step - Calculate density operators for states that Alice sends and Bob measures in the $\Z$ basis:}
	
	We first consider collective attacks - those are attacks where Eve prepares the same signal state each round.  We also consider the asymptotic scenario and so can assume that Alice and Bob have perfect statistics on the channel (i.e., we do not worry about finite sampling precision in this work).  Consider Eve's attack operator $U$ applied on each round of the protocol.
	First, let's define $U$'s action on a basis state $\ka \in \zbasis$ in the following way:
	\begin{align*}
		U\ka_T \ot \ket{\chi}_E = \sum_{b = 0}^{D - 1} \ket{b , \eba}_{TE},
	\end{align*}
	where $\ket{\chi}_E$ is a pure state in Eve's memory and states $\ket{e_i^j}$ are arbitrary states in Eve's memory. Note that, subscript $T$ and $E$ represent the transit and Eve's memory registers respectively.
	The unitarity of operator $U$ ensures that for any $\ka \in \zbasis$:
	\begin{align}
		\sum_{b = 0}^{D - 1} \bk{\eba}{\eba} = 1. \label{eq:zunitary}
	\end{align}
	When Alice decides a particular round to be a Key Round, she prepares an entangled state where each of the $\Z$ basis states is equally likely to appear. More precisely, Alice prepares a state $\ket{\psi}_{AT} \coloneqq \frac{1}{D} \sum_{a = 0}^{D - 1} \ket{a, a}_{AT}$ and then sends the transit part($T$) of the prepared state through the quantum channel, keeping her own register ($A$) to herself. Eve attacks this register $T$ with operator $U$ and the state evolves to the following:
	\begin{align}
		\left(\mbi \ot U \right)\ket{\psi}_{AT} \ot \ket{\chi}_E &= \left(\mbi \ot U \right)  \sqobd \sum_{a = 0}^{D - 1} \ket{a, a}_{AT} \ot \ket{\chi}_E \nonumber\\
		&= \sqobd \sum_{a = 0}^{D - 1} \ka_A \ot U \ka_T \ot \ket{\chi}_E \nonumber\\
		&= \sqobd \sum_{a = 0}^{D - 1} \ka_A \ot \sum_{b = 0}^{D - 1} \ket{b, \eba}_{BE} \coloneqq \ket{\psi}_{ABE}. \nonumber
	\end{align}
	In the last equality, we rename the transit register $T$ that arrives at Bob's lab to $B$. 
	
	Bob applies a $\Z$ basis measurement on register $B$, and the state collapses from this superposition to the states that agree with his measurement outcomes. After that, Alice also measures her subsystem in the $\Z$ basis and the state $\rho_{ABE}$, after their measurements, becomes the following mixed state:
	\begin{align}
		\rho_{A^ZB^ZE} = & \obd \sum_{a, b} \kb{a}{a}_A \ot \kb{b}{b}_{B} \ot \kb{\eba}{\eba}_{E}. \label{eq:sigabe}
	\end{align}
	Note that, this is the only case where a key bit could be generated. Because we want to bound the conditional entropy between Bob and Eve only, as we are using reverse-reconciliation \cite{qkd-survey-pirandola} as discussed earlier, we need to trace out Alice and get the joint cq-state between Bob and Eve. This results in:
	\begin{align}
		\Tr_A (\rho_{A^ZB^ZE}) = \rho_{B^ZE} = \obd \sum_{a, b} \kb{b}{b}_{B} \ot \kb{\eba}{\eba}_{E} \label{eq:rhobefirst}.
	\end{align}
	From this density operator, we need to calculate $H(B^Z|E)_\rho$. However, as discussed earlier, this is not straightforward. So we analyze the case when Alice sends a fixed $\X$-basis state $ \ket{x_0} = \sqobd \sum_{a = 0}^{D - 1} \ka_T$ and Bob measures in the $\Z$ basis. Then we shall show how to bound $H(B^Z|E)_\rho$ based on the information obtained from this case. We denote the subscript as the transit register $T$ because Alice sends this state in the quantum channel without keeping a copy in her lab. After Eve uses the unitary attack operator $U$ to attack this state, it evolves to:
	\begin{align}
		U \ket{x_0}_T \otimes \ket{\chi}_E &=  U \sqobd \sum_{a = 0}^{D - 1} \ka_T \otimes \ket{\chi}_E  \nonumber\\
		&= \sqobd \sum_{a = 0}^{D - 1} U\ka_T \otimes \ket{\chi}_E =\sqobd \sum_{a, b} \ket{b, e_b^a}_{BE} \coloneqq \kphibe \label{eq:phibe}.
	\end{align}
	The unitarity of $U$ ensures that,
	\begin{align}
		\bk{\phi}{\phi}_{BE} =  &\left(\sqobd \sum_{a, b} \bra{b, e_b^a}_{BE} \right) \left(\sqobd \sum_{a', b'} \ket{b', e_{b'}^{a'}}_{BE} \right) = \obd  \sum_{a, a', b} \braket{\eba | \ebap}_E  = 1. \label{eq:phibeunitarity}
	\end{align}
	Again, we rename the transit register $T$ to Bob's register $B$ as the state arrives at his lab. Now we take the outer product of state $\kphibe$ to obtain:
	\begin{align*}
		\sigma_{BE} \coloneqq \kb{\phi}{\phi}_{BE} = &\left( \sqobd \sum_{a, b} \ket{b, \eba}_{BE} \right)\left(\sqobd \sum_{a', b'} \bra{b',\ebpap}_{BE} \right) \\
		&=  \frac{1}{D} \sum_{a, a', b, b'} \kb{b}{b'}_B \ot \kb{\eba}{\ebpap}_E.
	\end{align*}
	Now, after receiving the state, he measures it in the $\Z$ basis. We need this density operator to perform continuity bound analysis in the third step of our analysis. After Bob measures this state in the $\Z$ basis, it becomes the following mixed state:
	\begin{align}
		\sigma_{B^ZE} = \obd \sum_{a, a', b} \kb{b}{b}_{B} \ot  \kb{\eba}{\ebap}_{E} \label{eq:sigbefirst}.
	\end{align}
	With these two density operators $\sigma_{B^ZE}$ and $\rho_{B^ZE}$ in our hands, where $\rho_{B^ZE}$ is needed to generate key bits and $\sigma_{B^ZE}$ to be used in continuity bound analysis, we can proceed to the next step.
	
	\subsection{Second Step - Calculate a lower bound on the key rate } Now we present the second part of the security analysis in the form of the following theorem. Here, $B^X$ is the random variable that represents Bob's $\X$ basis measurement outcomes, when Alice sends $\ket{x_0}$.  For $\modefull$, $B^X$ takes on $D$ possible values since $B$ is able to perform a full $\X$ basis measurement, while for $\modepartial$, $B^X$ takes only two values.
	\begin{theorem}\label{thm:keyratelemma}
		The key rate of protocol \eqref{prot:hd-3-State-BB84} when $\modefull$ is lower bounded by:
		\begin{align}
			K \ge \log_2D - \Delta - H(B^X)_\sigma - leak_{EC}, \label{eq:keyratethm-full}
		\end{align}
                and, when $\modepartial$, is lower bounded by:
                \begin{equation}\label{eq:keyratethm-partial}
                  K \ge \log_2D - \Delta - \frac{H_d(B^X)_\sigma}{\log_D2} - leak_{EC}.
                \end{equation}
              \end{theorem}
	where $\Delta = | \hbzerho - \hbzesig |$ based on the $\rho$ and $\sigma$ found in equations \eqref{eq:rhobefirst} and \eqref{eq:sigbefirst} respectively. Additionally, $leak_{EC}$ is the information leaked during error correction \cite{renner2005information}.
	\begin{proof}
		In the case when Alice sends a single $\X$ basis state which, after Bob's measurement, results in the density operator $\sigma_{BE}$, we know that by Equation \ref{eq:eu-1}, we have for $\modefull$:
		\begin{align}
			H(B^Z | E)_\sigma &+ H(B^X | A)_\sigma \ge \log(D) \nonumber\\
			\implies H(B^Z | E)_\sigma &\ge \log(D) - H(B^X | A)_\sigma \nonumber\\
			&\ge \log(D) - H(B^X)_\sigma \label{eq:rholowerberta},
		\end{align}
		where $H(B^X)_\sigma$ is easily estimated using the observed error rate in the Test Round case.  Then combining equation \eqref{eq:rholowerberta} with the definition of $\Delta$ as mentioned before, and remembering that $leak_{EC} = H(B^Z | A^Z)$ asymptotically, if they use an optimal error correction reconciliation protocol \cite{renner2005information}, Devetak-Winter's key rate equation from \cite{devetak2005distillation} finishes the proof for $\modefull$.  For $\modepartial$, the same arguments can be used, but with Equation \ref{eq:eu-2}, thus completing the proof.

              \end{proof}

              To calculate the keyrate, one must bound $\Delta$.  The following lemma proven in \cite{iqbal2020high} for an alternative protocol, helps with this as we shortly see:
	\begin{lemma}
		\label{lem:evesattackoperator}
		(Adopted from \cite{iqbal2020high}) Assuming Alice and Bob only use mutually unbiased bases for state encoding, it is to Eve's advantage to send an initial state satisfying the following orthogonality constraint:
		\begin{equation*}
			\braket{\eba | \ebap} = \begin{cases}
				p(b | a), & \text{ if }  a = a' \\
				0, & \text{otherwise},
			\end{cases}
		\end{equation*}
		where,  $p(b | a) = \bk{\eba}{\eba}$ denotes the probability of Bob's $\Z$ basis measurement outcome being a specific $\ket{b} \in \Z$ given that Alice sent a state $\ka \in \Z$.
	\end{lemma}
	\begin{proof}
		(See also \cite{iqbal2020high}).  Clearly if $a = a'$, then $\bk{\eba}{\eba} = p(b|a)$; thus the only thing to show is that it is to Eve's advantage to set her other ancilla vectors to be orthogonal.  Clearly, orthogonality can only help Eve as it allows her to more easily distinguish Alice or Bob's state.  The only time it may be useful to Eve to set these to be non-orthogonal is if, by doing so, she can make other vectors closer to orthogonal by decreasing the observed noise.  We therefore show that these inner products, when $a\ne a'$ do not contribute to any observed noise parameter and so Eve might as well set them to zero as they will not influence any  noise statistics, and orthogonal states will provide her with maximal information gain.
		
		Clearly, there are no terms of the form $\braket{\eba | \ebap}$ in the $\Z$ basis noise as we can see in equation \eqref{eq:sigabe}. Now, let's consider the $\X$ basis noise. When Alice sends a single $\X$ basis state $\ket{x_0} = \sqobd \sum_{a = 0}^{D - 1} \ka$ and Eve attacks, we get $\kphibe = \sqobd \sum_{a, b} \ket{b, e_b^a}_{BE}$ as shown in equation \eqref{eq:phibe}.
		In a noiseless scenario, if Bob measures in the $\X$ basis, he would receive the same $\ket{x_0}$ that Alice has sent. However, to measure how much disturbance is created by Eve, we rewrite $\kphibe$ in the $\X$ basis to get the probability of each of the outcomes in this basis. Each $\ket{b}$ in state $\kphibe$, is the following:
		\begin{align*}
			\ket{b} = \sum_i \beta_{b, i} \ket{x_x},
		\end{align*}
		where $\beta_{b, i} = \braket{b | x_i}$. Our choice of the $\X$ basis ensures that $\X$ and $\Z$ are mutually unbiased, meaning, $|\beta_{b, i}|^2 = \obd$ holds for any $b, i$ pair. Then we can write the state $\kphibe$ as:
		\begin{align*}
			\kphibe = &\sqobd \sum_{a, b, i} \beta_{b, i} \ket{x_i}_B  \ot \ket{e_b^a}_E, \\
			&=  \sum_i \ket{x_i}_B \ot \sqobd \sum_{a, b} \beta_{b, i} \ket{\eba}_E.
		\end{align*}
		In this case, if Bob measures this state in the $\X$ basis, the probability of any particular outcome $\ket{x_i} \in \X$, denoted as $p(x_i)$, is the following:
		\begin{align*}
			p(x_i) = &\left(\sqobd \sum_{a, b} \beta_{b, i}^\dagger \bra{\eba}_E \right) \left( \sqobd \sum_{a', b'} \beta_{b', i} \ket{\ebpap}_E \right) \\
			&= \obd \sum_{a, a', b, b'}  \beta_{b, i}^\dagger \beta_{b', i} \braket{\eba | \ebpap}_E \\
			&= \obd \sum_{a, a', b = b'}  \beta_{b, i}^\dagger \beta_{b, i} \braket{\eba | \ebap}_E +  \obd \sum_{a, a', b \ne b'}  \beta_{b, i}^\dagger \beta_{b', i} \braket{\eba | \ebpap}_E \\
			&= \frac{1}{D} \left( \sum_{a, a', b } \obd  \braket{\eba | \ebap}_E\right)  +  \obd \sum_{a, a', b \ne b'}  \beta_{b, i}^\dagger \beta_{b', i} \braket{\eba | \ebpap}_E \\
			&= \frac{1}{D} +  \obd \sum_{a, a', b \ne b'}  \beta_{b, i}^\dagger \beta_{b', i} \braket{\eba | \ebpap}_E,
		\end{align*}
		where the last equality follows using equation \eqref{eq:phibeunitarity}, which reminds us that terms of the form $\braket{\eba | \ebap}$ fall under the unitarity constraint on $U$. From this, we see that no terms of the form $\braket{\eba|\ebap}$ appear in any observed noise parameter even if Bob were able to make a full $\X$ basis measurement.   So, it does not hurt Eve if she decides to make terms like $\braket{\eba | \ebap}$ for $a \ne a'$ orthogonal, as these terms do not contribute to either $\Z$ or $\X$ basis noise. So, Eve might as well set them to zero.
	\end{proof}

	\paragraph{Calculating $\Delta$ in depolarizing channel:}

        First, we show an analytical expression for $\Delta$ assuming depolarization channels.  Later, we will show how one may bound $\Delta$ for arbitrary channels.
	
	To bound $\Delta$ we take advantage of the continuity of von Neumann entropy; namely that $\Delta$ can be upper-bounded as a function of the trace distance between $\rho$ and $\sigma$.  In particular, \cite{winter2016tight} bounds the absolute difference of conditional entropies of two bipartite cq-states, in this case, $\rho_{B^ZE}$ and $\sigma_{B^ZE}$, as a function of their trace distance. We have,
	\begin{align}
		\Delta = | H(B^Z|E)_\rho - \hbzesig | \le \epsilon \log |B^Z| + (1 + \epsilon) h\left(\frac{\epsilon}{1 + \epsilon}\right), \label{eq:wintercontinuityExtra}
	\end{align}
	where $\epsilon \ge \half \| \sigma_{B^ZE} - \rho_{B^ZE} \|_1$ is the upper bound of trace distance between systems $\sigma_{B^ZE}$ and $\rho_{B^ZE}$ and $|B^Z|$ is the size of the set of outcomes from Bob's $\Z$ basis measurement, in our case, which is simply $D$.
	
	In this section, we show how to bound the trace distance as a function only of the noise in the channel and the dimension of each system $D$.  To do so, we use methods we developed in \cite{iqbal2020high}.  While our bound only applies to a depolarization channel (the most common form of noise considered in theoretical QKD analysis), our methods may be extendable to other channels.  Note that this derivation below is the only place in our proof where we assume a depolarization channel.  A depolarizing channel, denoted here by $\mathcal{E}_q$ with noise parameter $q$, acts on a quantum state $\omega$ with dimension $D$ in the following way:
	\begin{equation}\label{eq:depolar}
		\mathcal{E}_q(\omega) = \left(1-\frac{D}{D-1}q\right)\omega + \frac{q}{D-1}\mathbb{I},
	\end{equation}
	where $\mathbb{I}$ is the identity operator of dimension $D$.
	
	
	Using basic properties of trace distance we have:
	\begin{align}
		\epsilon &= \half \left\| \sigma_{B^ZE} - \rho_{B^ZE} \right\|_1 \nonumber \\
		&= \frac{1}{2D} \left\| \sum_{a, a', b} \kb{b}{b}_B \ot  \kb{\eba}{\ebap}_E - \sum_{a, b} \kb{b}{b}_B \ot \kb{\eba}{\eba}_E \right\|_1 \nonumber \\
		&= \frac{1}{2D} \left\| \sum_{b = 0}^{D - 1} \kb{b}{b}_B \ot \left(\sum_{a, a'} \kb{\eba}{\ebap}_E - \sum_{a = 0}^{D - 1} \kb{\eba}{\eba} \right) \right\|_1 \nonumber \\
		&= \frac{1}{2D}  \sum_{b = 0}^{D - 1} \left \| \sum_{a, a'} \kb{\eba}{\ebap} - \sum_{a = 0}^{D - 1} \kb{\eba}{\eba} \right \|_1 \nonumber \\
		&= \frac{1}{2D}  \sum_{b = 0}^{D - 1} \left\| \sum_{a \ne a'} \kb{\eba}{\ebap} \right\|_1 \label{eq:tracedistupb},
	\end{align}
	where the second to last equality follows because $(\sigma_{B^ZE} - \rho_{B^ZE})$ is a cq-state and block-diagonal, and the eigenvalues of each block for an individual $\ket{b} \in \{ \ket{0}, \ket{1}, ..., \ket{D - 1} \}$ can be found individually. Taking advantage of Lemma \ref{lem:evesattackoperator}, along with the fact that trace distance is invariant to changes in basis, we find:
	\begin{align}
		\epsilon &= \frac{1}{2D} \sum_{b = 0}^{D - 1} \left\| \sum_{a \ne a'} \kb{\eba}{\ebap} \right \|_1 \notag\\
		&= \frac{1}{2D} \sum_{b = 0}^{D - 1} \left\| \sum_{a \ne a'} \sqrt{p(b | a)p(b | a')} \kb{a}{a'} \right\|_1,\label{eq:tr-dist-general}
	\end{align}
	where, for the last equality, we have used the computational basis, meaning,  $a, a' \in \Z$, remembering that the trace norm is invariant under the change of basis. Now, we denote $X_b \coloneqq \sum_{a \ne a'} \sqrt{p(b | a)p(b | a')} \kb{a}{a'} $. Clearly, $X_b$ is a Hermitian operator. So we can write it as $X_b = \sum_i \lambda_i \kb{v_i}{v_i}$, where $\lambda_i$ and $v_i$ are it's eigenvalues and orthogonal eigenvectors, respectively. Its trace norm is the sum of the absolute values of its eigenvalues, i.e.,  $\|X_b\|_1 = \sum_i |\lambda_i|$. We can write an eigenvector $\ket{v}$ of $X_b$ in an orthonormal basis $\{\ket{i}\}_{i = 0}^{D - 1}$ like: $\ket{v} = \sum_i x_i \ket{i}$. Now, if we apply the operator $X_b$ on one of its eigenvectors, we get the following:
	\begin{align}
		X_b \ket{v} &= \sum_{a \ne a'} \sqrt{p(b | a)p(b | a')} \kb{a}{a'} \left(\sum_i x_i \ket{i}\right) \nonumber\\
		&= \sum_{a \ne a', i} \sqrt{p(b | a)p(b | a') } x_i  \ka\bk{a'}{i} \nonumber\\
		&= \sum_{a, i \ne a}  \sqrt{p(b | a)p(b | i) } x_{i} \ka \label{eq:xbv}
	\end{align}
	There are two families of equations in the equation above for pairs of $a, b$, i.e., when $a = b$ and $a \ne b$. Lets consider the first case when $a = b$, and keep in mind that $X_b \ket{v} = \lambda \ket{v}$ for some particular eigenvector $\ket{v} = \sum_i x_i \ket{i}$, now it holds that:
	\begin{align}
		\sum_{ \substack{i = 0  \\ i \ne a} }^{D - 1}  \sqrt{p(b | b)p(b | i) }x_{i} &= \lambda x_b, \nonumber\\
		\implies \sum_{ \substack{i = 0  \\ i \ne a} }^{D - 1}  \sqrt{\alpha \beta }x_{i} &= \lambda x_b,
		\label{eq:aeqb}
	\end{align}
	where we define $\alpha  = p(b | b) = 1 - q$ and $\beta = p(b | i) = \frac{q}{D - 1}$. Note that it is starting here that we use the assumption of a depolarization channel.  Then in the second case, when $a \ne b$, we have two more cases for $i = b$ and $i \ne b$. So we get:
	\begin{align}
		\sqrt{p(a | a)p(b | a)} x_b + \sum_{ \substack{i \ne a \\ i \ne b}}\sqrt{p(b | a)p(b | i)} x_i = \lambda x_a \nonumber \\
		\implies \sqrt{\alpha \beta} x_b + \sum_{ \substack{i \ne a \\ i \ne b}}\beta  x_i = \lambda x_a, \label{eq:aneqb}
	\end{align}
	where, we note that, $p(a | a) = 1 - q$ and $p(b | a) = \frac{q}{D - 1}$. Note that, here we are considering $0 < q < 1 - \obd$, because if $q = 0$ then there is no noise in the channel and we are done. If $q = 1 - \obd$, then $\alpha = \beta = \obd$ and we will consider the channel as too noisy for Alice and Bob to execute the protocol and so they must abort. Let us consider index $k \ne k', k \ne b, k' \ne b$,  such that $x_k \ne x_k'$. Then from equation \eqref{eq:aneqb}, we get, when $a = k$ and $a = k'$ respectively:
	\begin{align*}
		\sqrt{\alpha \beta} x_b + \sum_{ \substack{i \ne k \\ i \ne b}}\beta  x_i &= \lambda x_k \\
		\sqrt{\alpha \beta} x_b + \sum_{ \substack{i \ne k' \\ i \ne b}}\beta  x_i &= \lambda x_{k'}.
	\end{align*}
	Subtracting these two equations from each other yields an equation for one of the eigenvalues of $X_b$:
	\begin{align*}
		\beta x_{k'} - \beta x_k &= \lambda (x_k - x_{k'}), \\
		\Rightarrow \lambda &= \frac{\beta (x_{k'} - x_k)}{x_k - x_{k'}} = -\beta,
	\end{align*}
	where the last equality is possible because of our assumption, $x_k \ne x_{k'}$. We claim that the geometric multiplicity of this eigenvalue is $D - 2$. To see this, notice that by choosing a suitable basis, we can write the operator $X_b - \lambda \mathbb{I}$ as:
	\begin{align*}
		\begin{bmatrix}
			-\lambda & \sqrt{\ab} & \sqrt{\ab}  & \ldots & \sqrt{\ab} \\
			\sqrt{\ab} & -\lambda & \beta & \ldots &  \beta \\
			\sqrt{\ab} & \beta & -\lambda & \ldots & \beta \\
			\vdots & \vdots & \vdots & \ddots  & \vdots \\
			\sqrt{\ab} & \beta & \beta & \ldots & -\lambda
		\end{bmatrix}.
	\end{align*}
	We use the eigenvalue $\lambda = -\beta$ and see that there are only two linearly independent rows in $X_b - (-\beta)\mathbb{I} = X_b + \beta\mathbb{I}$. So the rank of $X_b + \beta\mathbb{I}$ is at most two, in this case, it is exactly two. It follows that the geometric multiplicity of $\lambda = -\beta$ is $D - 2$. To find the remaining two eigenvalues of $X_b$, now we consider the second case when there does not exist $k\ne k'$ ($k \ne b, k' \ne b$) such that $x_k \ne x_{k'}$. In this case we have $x_k = x_{k'} = x$. Then from equation \eqref{eq:aeqb} we get:
	\begin{align}
		\sqrt{\alpha \beta} (D - 1) x = \lambda x_b \implies x_b = \frac{\sqrt{\alpha \beta} (D -1) x}{\lambda} \label{eq:smallxbvalue}.
	\end{align}
	Substituting the value of $x_b$ from the above equation \eqref{eq:smallxbvalue} in equation \eqref{eq:aneqb} we get:
	\begin{align}
		\sqrt{\alpha \beta} \left( \frac{\sqrt{\alpha \beta} (D -1) x}{\lambda} \right) + \beta (D - 2) x &= \lambda x \nonumber\\
		\Rightarrow \frac{\alpha \beta (D - 1)}{\lambda}+ \beta (D - 2) &= \lambda \nonumber \\
		\Rightarrow \lambda^2 - \lambda \beta(D - 2) - (D - 1) \alpha \beta &= 0 \label{eq:lambeigeqn}.
	\end{align}
	Solving the above equation \eqref{eq:lambeigeqn}, we get the remaining two eigenvalues of $X_b$:
	\begin{align*}
		\lambda^\pm = \half \left( (D - 2) \beta \pm \sqrt{\beta} \sqrt{(D - 1) (4 \alpha - 4 \beta) + \beta D^2 } \right).
	\end{align*}
	So the trace distance between $\rho_{B^ZE}$ and $\sigma_{B^ZE}$ as defined in equation \eqref{eq:tracedistupb} is then bounded by:
	\begin{align}
		\epsilon \le \frac{1}{2D} \sum_{b = 0}^{D - 1}  \left \| \sum_{a \ne a'} \kb{\eba}{\ebap} \right \|_1 = \frac{1}{2D} \sum_{b = 0}^{D - 1} \left((D - 2)\left|-\frac{q}{D - 1}\right| + \left|\lambda^+ \right| + \left|\lambda^- \right| \right). \label{eq:epsilonupb}
	\end{align}
	Note that, $\epsilon$ only depends on the channel noise parameter $q$ and the dimension of the Hilbert spaces, $D$, used by Alice and Bob. With this value of $\epsilon$ calculated, we can easily get $\Delta$ as a function of only noise parameter $q$ and dimension $D$.

        \paragraph{Bounding $\Delta$ in the General Case: } Much of our analysis above was channel independent.  The only time we needed to assume a depolarization channel was in the derivation of Equation \eqref{eq:epsilonupb}.  In fact, the only absolute requirement is the symmetric nature of the noise $p(b|a)$, namely that $p(b|b) = 1-q$ and $p(b|a) = p(b'|a)$ for all $b$ and $b'$ that are not equal to $a$.  For non-depolarization channels, however, Theorem \ref{thm:keyratelemma} and Lemma \ref{lem:evesattackoperator} both still apply.  The challenge is in deriving a general expression for $\epsilon$ needed for the continuity bound.  However, if it is found that the channel does not satisfy the symmetry assumption, alternative methods may be used.  In particular, since Lemma \eqref{lem:evesattackoperator} still applies, Equation \ref{eq:tr-dist-general} is still valid for any channel.  One may then numerically evaluate this expression, finding the eigenvalues numerically, for the given channel (since each trace distance in that sum uses the computational basis this is not a difficult computation and does not require any numerical optimization).

        In detail, the following algorithm may be used to upper-bound $\epsilon$, which is needed to upper-bound $\Delta$:
        \begin{enumerate}
          \item Set a variable $td = 0$
        \item For each $b = 0, 1, \cdots D-1$ Do:
          \begin{enumerate}
          \item Set $M$ to be a zero-matrix of size $D\times D$.
          \item For each $a, a' = 0, 1, \cdots, D-1$ with $a \ne a'$ Do:
            \begin{enumerate}
              \item $M = M + \sqrt{p(b|a)p(b|a')}\kb{a}{a'}$
            \end{enumerate}
          \item Compute the eigenvalues $\{\lambda_1, \cdots\}$ of $M$
            \item $td = td + \sum_i|\lambda_i|$
            \end{enumerate}
            \item Return the Trace Distance: $td/(2D)$.
            \end{enumerate}

            In general, we found that our algorithm runs significantly faster than the numerical approach used in prior work \cite{islam2018securing}.  Indeed, for the amplitude damping channel, evaluated below, the difference in running time on a standard desktop computer was order of magnitude faster for our approach (taking seconds as opposed to hours).
        
	\paragraph{Key Rate calculation:}
	To calculate the key rate using Equation \ref{eq:winterkeyrate}, we need expressions for two more terms, namely, $H(B^X)$ and the number of classical bits revealed during error correction, $leak_{EC} = H(B^Z | A^Z)$.  Both of these are observable parameters; for evaluation purposes, we will simulate their expected values under a depolarization channel.  First, $H(B^X)$ is found to be (assuming a depolarization channel):
	\begin{align}
		H(B^X) \le q \log(D - 1) + h(q). \label{eqn:hbzkeyrateterm}
	\end{align}
	Note that we are upper-bounding this quantity as Bob cannot distinguish all $\X$ basis states, only $\ket{x_0}$ from the others.
	
	The second term $H(B^Z | A^Z)$ is calculated below. Note that, in the following, we use $p_{a, b}$ as the probability of Alice sending $\ka$ and Bob measuring $\ket{b}$ in $\Z$ basis, and use $p_a$ as the probability of Alice sending $\ka$. We see that,
	\begin{align}
		H(B^Z | A^Z) &= H(B^Z A^Z) - H(A^Z)  \nonumber\\
		&= -\sum_{a, b} p_{a, b} \log p_{a, b}  + \sum_{a = 0}^{D - 1} p_a \log p_a \nonumber\\
		&= -\sum_{a, b} \frac{p(b|a)}{D} \log \frac{p(b|a)}{D} - \log(D) \nonumber\\
		&= - \obd \sum_{a, b} p(b|a) \log \frac{p(b|a)}{D} - \log(D) \nonumber \\
		&= - \obd \sum_{a = b} p(a|a) \log \frac{p(a|a)}{D} - \obd \sum_{a \ne b} p(b|a) \log \frac{p(b|a)}{D} - \log(D) \nonumber\\
		&= - \obd \sum_{a = b} (1 - q) \left( \log (1 - q) - \log(D) \right)  -  \obd \sum_{a \ne b} \frac{q}{D - 1} \left( \log \frac{q}{D - 1} - \log(D) \right)   - \log(D) \nonumber\\
		&= - (1 - q) \left( \log (1 - q) - \log(D) \right)  - q \left( \log \frac{q}{D - 1} - \log(D) \right)   - \log(D) \nonumber\\
		&= q \log(D - 1) + h(q).  \label{eq:hbzaz}
	\end{align}
	At this point, we finally get an analytical lower bound on the key rate in equation \eqref{eq:keyratethm-full} after considering equations \eqref{eq:wintercontinuityExtra}, \eqref{eqn:hbzkeyrateterm}, and \eqref{eq:hbzaz}. Now it reads:
	\begin{align}
		K &\ge \log(D) - 2q \log(D - 1) - 2h(q) -\epsilon \log \left( \left|B^Z \right|\right) - (1 + \epsilon) h\left(\frac{\epsilon}{1 + \epsilon}\right). \label{eq:keyratefinal}
	\end{align}
        A similar bound for $\modepartial$ is found using Equation \eqref{eq:keyratethm-partial}.
        
	Before we move on to the evaluation section of our work, below we present a slightly improved lower bound of the key rate $K$ for $\modefull$.

	\subsection{Improved Key Rate Bound}
	Note that the key rate in equation \eqref{eq:keyratefinal} applies to arbitrary dimensions. In a specific case when $D = 2$ and $0 \le q \le .1464$, we can obtain a slightly improved key rate which we show below. This new key rate is based on lemma \eqref{lem:ourlemma} (stated formally in the Appendix, and also proven there). This lemma is an improved continuity bound for $\Delta$ that applies in this specific case, as opposed to the bound by Winter \cite{winter2016tight} which, though higher, is more general (i.e., our bound produces a better result for the specific cases we consider here, however Winter's bound is more general and can be applied to any quantum states). The inspiration for this lemma comes from the work by Wilde, who in \cite{wilde2020optimal} asked whether the following could hold:
	\begin{align}
		|H(B^Z|E)_\rho - H(B^Z|E)_\sigma| \stackrel{?}{\le} \epsilon \log(D_B - 1) + h(\epsilon).
	\end{align}
	If proven true, this conjecture would really help in providing more optimistic bounds for key rates where such continuity bounds are used to prove security. While we could not prove the exact conjecture, we do show in lemma \eqref{lem:ourlemma} that a slightly weaker lower bound holds (though better than prior proven work for the particular type of states we consider here). Namely, for $q \le 14.64\%$ and $D=2$, it holds that:
	%
		\begin{align*}
			\left|H(B^Z|E)_\rho - H(B^Z|E)_\sigma \right|\le h(1 - q -\mq)
	\end{align*}
We present the proof of this in the appendix section. Based on lemma \eqref{lem:ourlemma}, we get a slightly improved lower bound for the key rate $K$.
Here we use the notation for key rate $K'$ to mean that we are using our new continuity bound from lemma \eqref{lem:ourlemma} and the dimension and the noise is restricted to a specific case mentioned above (namely, $D=2$ and depolarization channels). This new key rate bound reads:
\begin{align}
	K' &\ge 1 - 2 h(q) - h\left(1 - q - \sqrt{q(1 - q)}\right). \label{eq:keyratenewsimple}
\end{align}
Now we evaluate these two key rates in the next section.

\subsection{General Attacks}

While the above analyzed collective attacks, the security analysis may be promoted in a straightforward manner to general attacks.  First, observe that the protocol may be reduced to an equivalent entanglement-based version in the following way.  First, Eve prepares an arbitrary state $\ket{\psi}_{ABE}$ which we may write without loss of generality as:
\[
\ket{\psi}_{ABE} = \frac{1}{\sqrt{D}}\sum_{a,b}\ket{a,b,e_b^a}.
\]
Note that, if we assume Alice observes $\ket{a}$ with probability $1/D$ (which may be enforced), this then implies $\sum_{b=0}^{D-1}\braket{e_b^a|e_b^a} = 1$ (which we assumed in the previous section).  Now, on a Key Round, Alice and Bob will measure as normal.  On a Test Round, Alice can measure in the $X$ basis and reject the signal if she does not observe $\ket{x_0}$.  It is not difficult to see that, conditioning on Alice observing $\ket{x_0}$, this will disentangle her system with Bob and Eve's in the same manner as if she had initially sent $\ket{x_0}$.  Of course, such an entanglement protocol is very wasteful - many rounds are discarded - nonetheless conditioning on not rejecting a signal, the resulting quantum states are identical to the prepare and measure version analyzed earlier.  Thus security there implies security of the entanglement-based version and vice versa.  Finally, de Finetti style arguments \cite{konig2005finetti} may be used to promote security to general attacks, thus concluding the proof.

\section{Evaluation}

In the following, we evaluate our key rate bound in equation \eqref{eq:keyratefinal} in a depolarizing channel and compare it with prior work in \cite{islam2018securing}. Then we further evaluate our protocol in the amplitude damping channels.

\subsection{Depolarizing Channel}
The depolarization channel is modeled in Equation \ref{eq:depolar}. Because prior work methods are computationally intensive to replicate, we only present the comparison with our analysis for key rates for up to $D = 8$.   For the depolarization channel, our results and comparisons are presented in figure \eqref{fig:fig1oursvsnurulsd27} by evaluating our key rate from equation \eqref{eq:keyratefinal}. Then, in figure \eqref{fig:fig2depolfull}, we present our key rates of $\modefull$ for higher dimensions (though without a comparison).
\begin{figure}
	\centering
	\includegraphics[width=\linewidth]{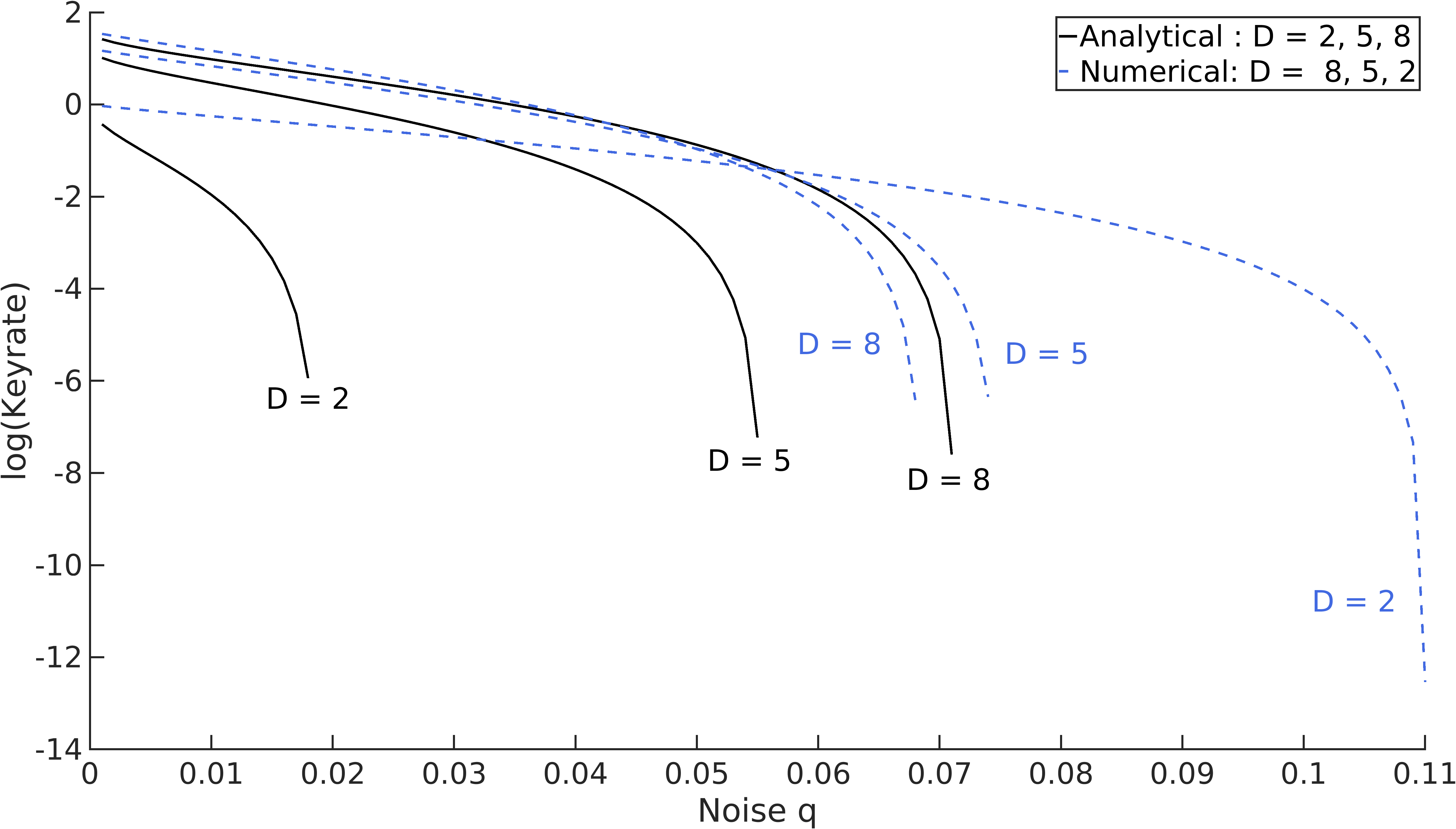}
	\caption{Comparison of our analysis and the numerical method from \cite{islam2018securing}. Notice the decreasing trend of key rates in the numerical approach(dotted lines), and the opposite in our case (solid lines). Here we evaluate our protocol in the $\modefull$ case, which is the analogous case of \cite{islam2018securing}.}
	\label{fig:fig1oursvsnurulsd27}
\end{figure}
\begin{figure}
	\centering
	\includegraphics[width=\linewidth]{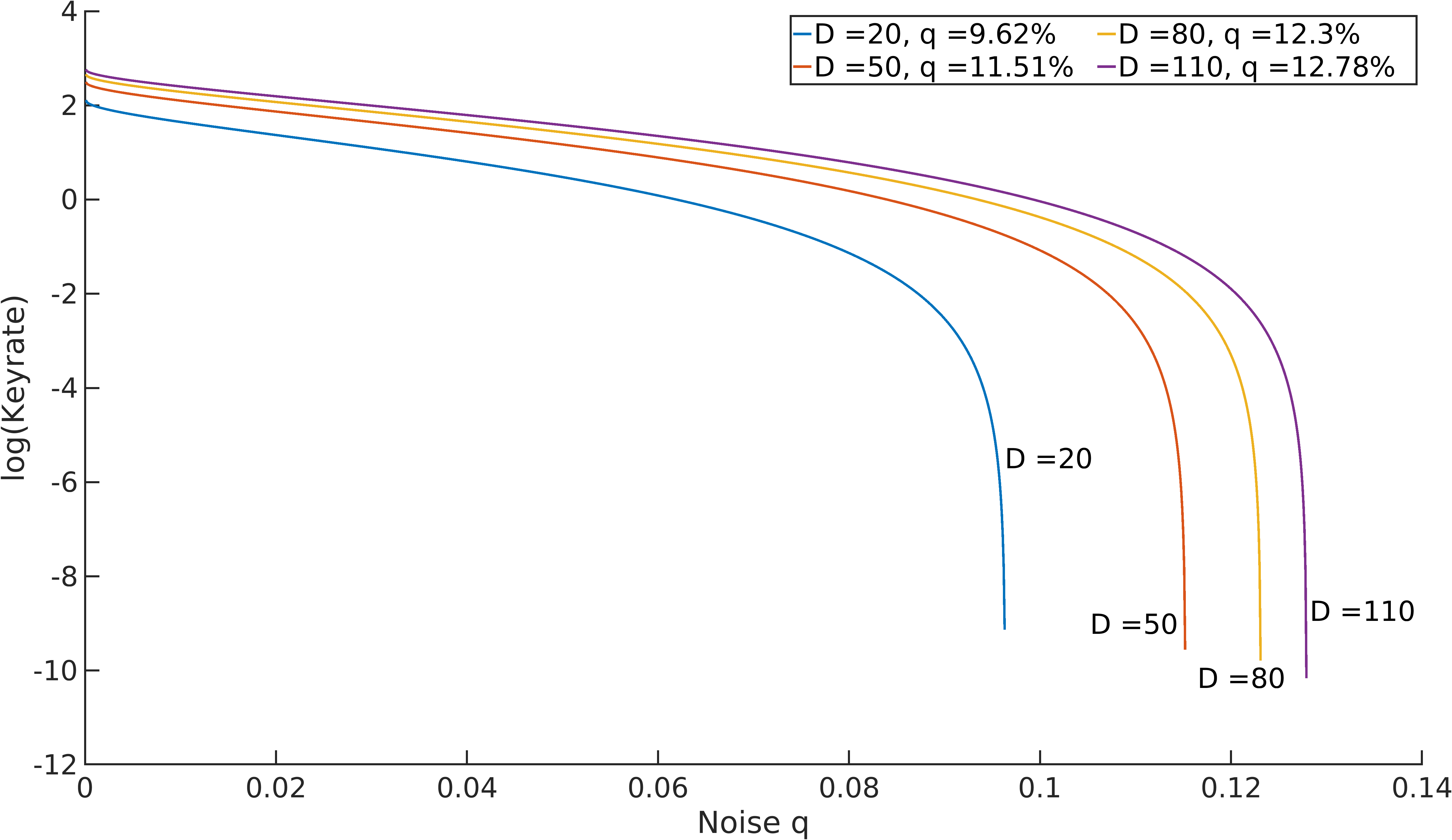}
	\caption{Noise tolerance of HD-3-State-BB84 in $\modefull$  from dimension 10 to 110 in our analysis for the depolarizing channel. }
	\label{fig:fig2depolfull}
\end{figure}

\begin{figure}
	\centering
	\includegraphics[width=\linewidth]{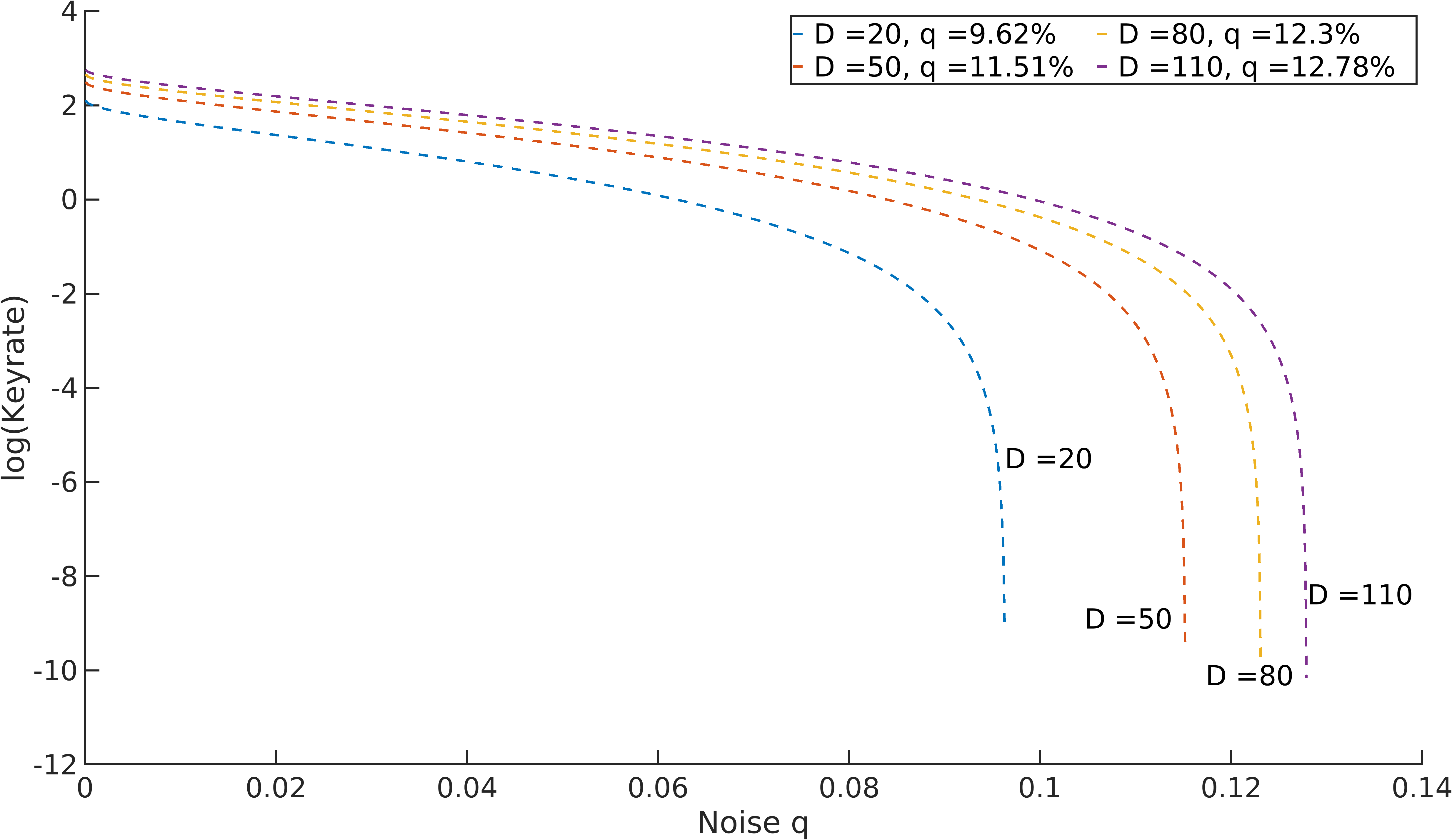}
	\caption{Noise tolerance of HD-3-State-BB84 in $\modepartial$  from dimension 10 to 110 in our analysis for the depolarizing channel. }
	\label{fig:fig2depolpartial}
\end{figure}

\begin{figure}
	\centering
	\includegraphics[width=\linewidth]{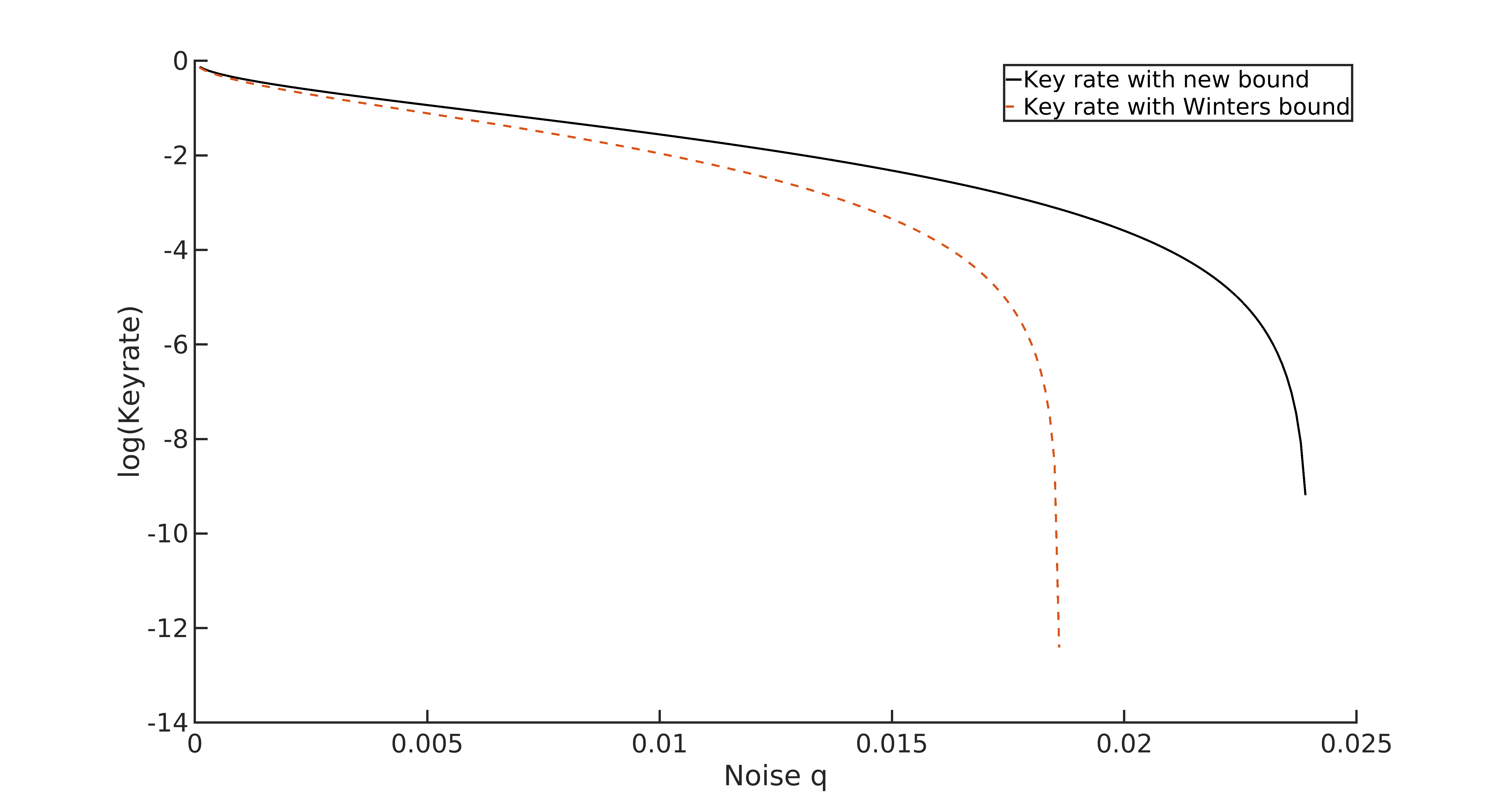}
	\caption{Comparison of the noise tolerances of HD-3-State-BB84 in dimension $2$ and in $\modefull$, considering our lemma \ref{lem:ourlemma} and Winter's bound\cite{winter2016tight} for conditional quantum entropies.}
	\label{fig:fig3oursVsWinterContinuity}
\end{figure}

In figure \eqref{fig:fig1oursvsnurulsd27}, it can be seen that the noise tolerance determined using the previous method in \cite{islam2018securing} actually goes down with increasing dimensions when only one monitoring basis is used. For example, it is $7.45\%$ for $D = 3$ and $7.28\%$ for $D = 6$. As indicated in their work, this may be attributed to the quick rise of the optimal phase error rate produced by the optimization algorithm with increasing dimensions. For example, the phase error rate is $27.68\%$ for $D = 3$ and $45.04\%$ for $D = 6$. This effect is more pronounced when one compares the result in the case of $D = 2$, where their analysis performs best in terms of noise tolerances compared to dimensions $D > 2$. In contrast to their result, in figure \eqref{fig:fig1oursvsnurulsd27}, we see that our analysis produces better noise tolerances as dimensions increase, as one may expect based on other HD-QKD protocols. As an example, in our analysis, the noise tolerance is $3.541\%$ for $D = 3$ and $6.169\%$ for $D = 6$. In dimension $D = 9$, we see that our noise tolerance is $7.482\%$ compared to $6.63\%$ in prior work \cite{islam2018securing}. Note that both our results and previous work results are lower bounds so there is no contradiction here. In fact, users can simply take the maximum of our work and prior work.

The decreasing trend in noise tolerance from prior work with increasing dimensions when only one monitoring state is used, and the opposite case in our analysis, leads us to suspect that in even higher dimensions $D > 9$, our analysis would continue to produce better noise tolerances. Indeed, this is shown in Figure \ref{fig:fig2depolfull}. Note that a compelling theoretical study would be finding out what happens when Alice sends any number of (1 to $D - 1$) monitoring basis states in the protocol. Although an advantage is shown numerically in prior work \cite{islam2018securing} for sending more monitoring bases, an analytical proof would be an enlightening future work.  Thus, while our result under-performs prior work for small dimensions, it greatly outperforms prior work once the dimension increases and also proves for the first time that HD states do, in fact, benefits this HD-three-state protocol (prior work could not confirm this as discussed).


Comparing these two modes, we notice that the performance of $\modepartial$ is equal to the performance of $\modepartial$ under the depolarizing channel.  Considering the proof of Equation \eqref{eq:eu-2} from \cite{krawec2020new} assumes the ``worst case'' distribution on Bob's measurement, this is not surprising. The case will be different for other channels; we will see this fact later when we evaluate the protocol under the amplitude damping channel. It is worth noting that in certain channels, like the depolarizing channel here, even with $\modepartial$, we can achieve same performance level as $\modefull$, where more resources are required. 


It can be said from these graphs that our analysis clearly demonstrates the advantage of using high-dimensional resources as well as the feasibility of obtaining competitive performance even when one uses much fewer quantum resources.
Finally, in figure \eqref{fig:fig3oursVsWinterContinuity}, we use our new continuity bound (Lemma \ref{lem:ourlemma}) instead of using Winter's continuity bound \cite{winter2016tight}, and evaluate key rate equation \eqref{eq:keyratenewsimple} to compare it with equation \eqref{eq:keyratefinal} by setting $D = 2$. We see that the noise tolerance increases from $1.85\%$ to $2.39\%$ demonstrating its utility. Notably, this improvement holds for both modes of our protocol, and we only show this for $\modefull$. We again emphasize that this new bound is limited to $D = 2, 0 \le q \le .1464$, and applicable in the depolarizing channel.  While our continuity bound is only applicable to dimension two currently, future work may be able to improve this.

\subsection{Amplitude Damping Channel}
We further evaluate our analysis in another widely used noise model, namely, the amplitude damping channel \cite{fonseca2019high}. This channel can be described by the following Kraus operators:

\begin{equation*}
	E_0 = \left(\begin{array}{ccccc}
		1 & 0 & 0 &\cdots & 0\\
		0 & \sqrt{{1-p}} & 0 &\cdots & 0\\
		0 & 0 & \sqrt{{1-p}} &\cdots & 0\\
		\vdots & \vdots & \vdots & \ddots & \vdots\\
		0 & 0 & 0 &\cdots & \sqrt{{1-p}}\end{array}\right) ,
\end{equation*}

\begin{equation*}
	E_1 = \left(\begin{array}{ccccc}
		0 & \sqrt{{p}} & 0 &\cdots & 0\\
		0 & 0 & 0 &\cdots & 0\\
		0 & 0 & 0 &\cdots & 0\\
		\vdots & \vdots & \vdots & \ddots & \vdots\\
		0 & 0 & 0 &\cdots & 0\end{array}\right),
	\cdots ,
	E_{D-1} = \left(\begin{array}{ccccc}
		0 & 0& 0 &\cdots & \sqrt{{p}} \\
		0 &0& 0 &\cdots &0\\
		0 & 0 & 0 &\cdots &0\\
		\vdots & \vdots & \vdots & \ddots & \vdots\\
		0 & 0 & 0 &\cdots & 0\end{array}\right).
\end{equation*}
from which the channel maps a density operator $\rho$ to $\sum_iE_i\rho E_i^\dagger$. For the $\modefull$ case of our protocol, where Bob performs a full $X$ basis measurement, we find the following key rates in higher dimensions:
\begin{figure}[H]
	\centering
	\includegraphics[width=\linewidth]{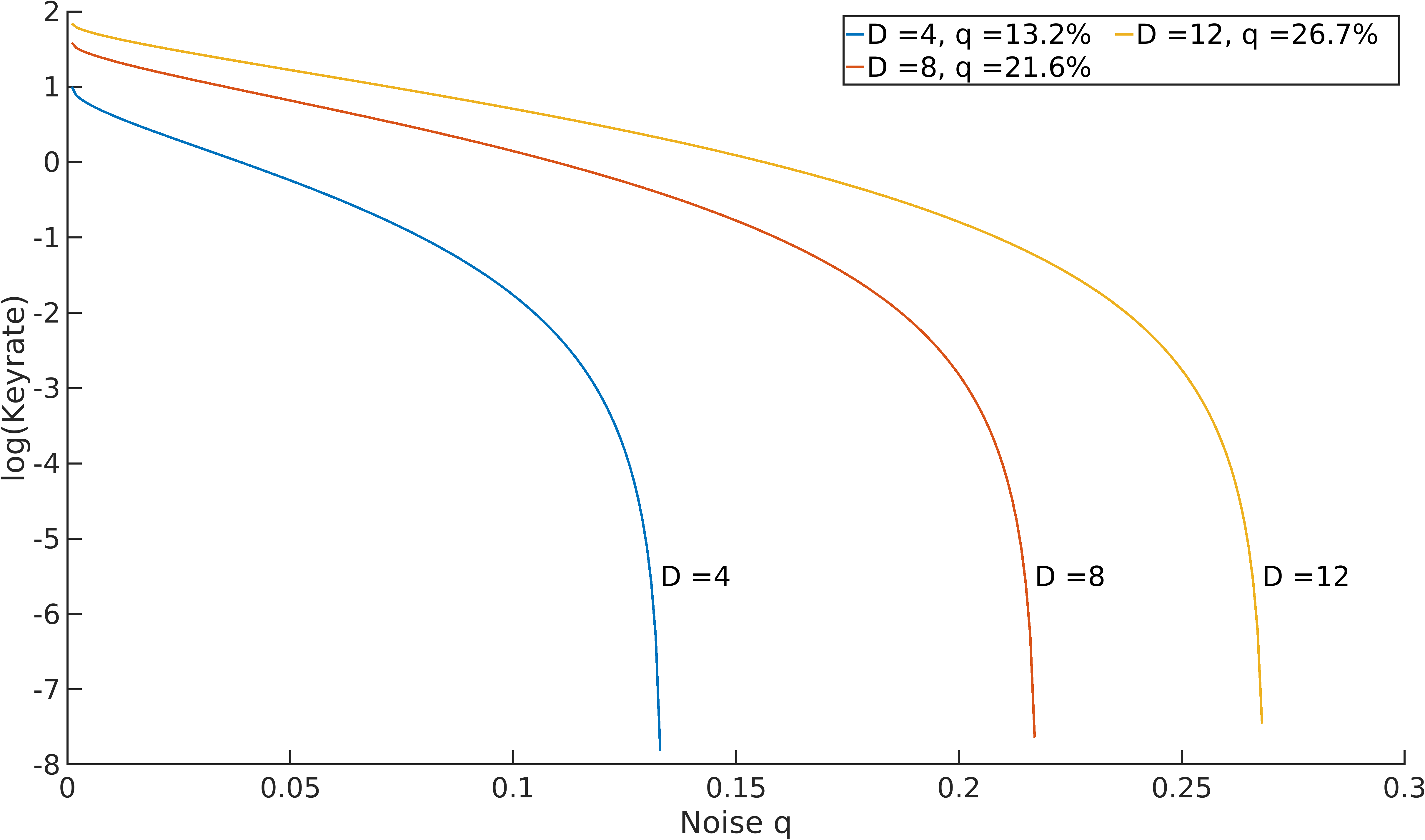}
	\caption{Key rates for HD-3-state-BB84 protocol in $\modefull$ when the amplitude damping channel is used. We consider dimensions $D = 4, 8, 12$ here. }
	\label{fig:fig2oursampdampprot1}
\end{figure}
We present the evaluation of our analysis for this channel in figure \eqref{fig:fig2oursampdampprot1}, and \eqref{fig:fig2ampdamppartial}. When we consider the $\modepartial$, shown in \eqref{fig:fig2ampdamppartial} for HD-3-State-BB84 protocol in this channel, where Bob needs much less resources to implement his measurement apparatus, we see that the noise tolerance is competitive for smaller dimensions compared to $\modefull$, shown in figure \eqref{fig:fig2oursampdampprot1}. For example, it is $12.3\%$ for $D = 4$ in $\modepartial$ compared to $13.2\%$ for $\modefull$. However, as the dimension increases, $\modefull$ outperforms $\modepartial$ more significantly. From these observations in this channel, we draw similar conclusions as in the depolarizing channel. That is, high-dimensional resources do offer better performance for HD-3-State-BB84 protocol. 

\begin{figure}[H]
	\centering
	\includegraphics[width=\linewidth]{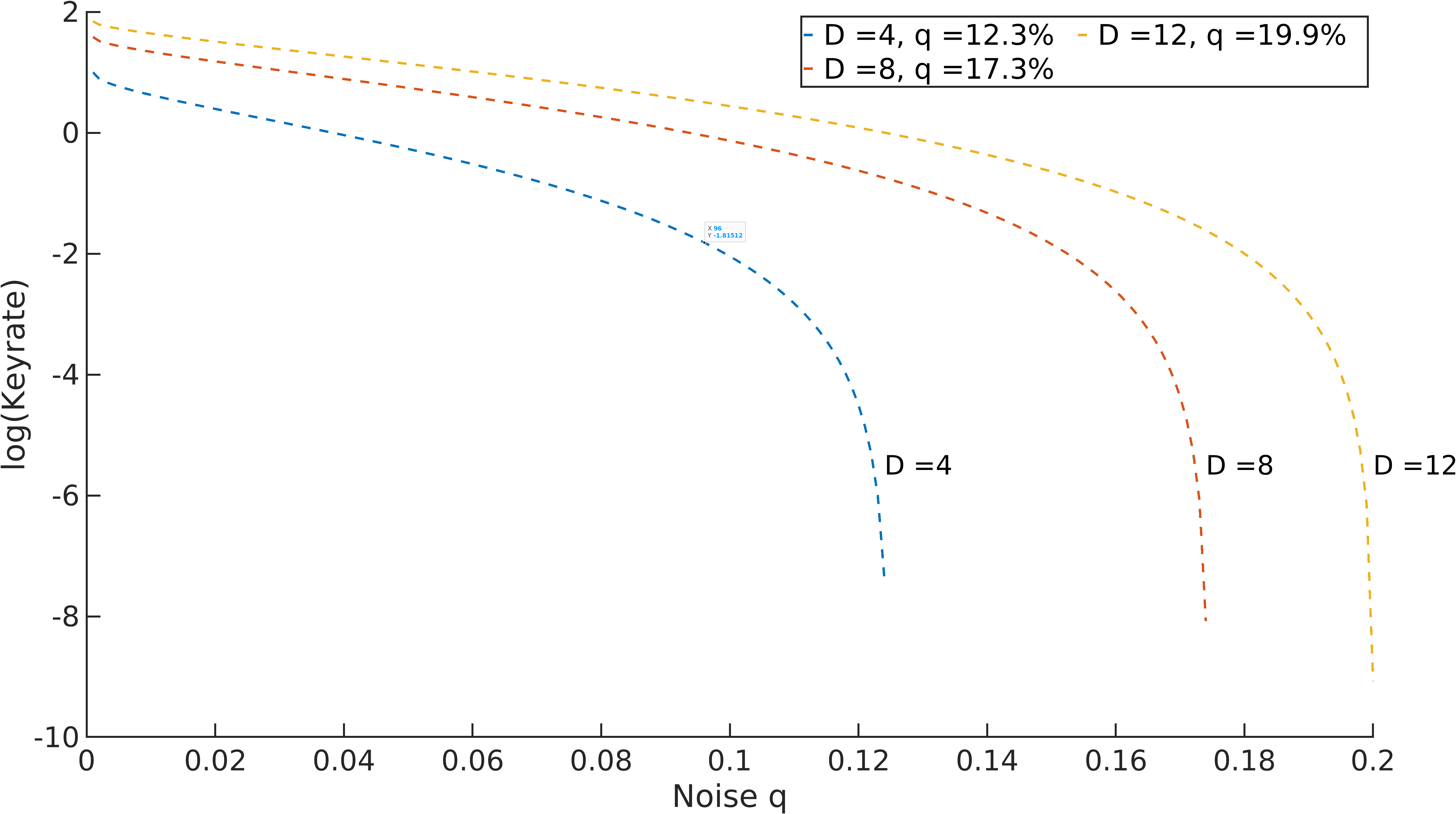}
	\caption{Key rates for HD-3-state-BB84-2 protocol, where the two-outcome POVM EU relation is used, in the amplitude damping channel for dimensions D = 4, 8, 12.}
	\label{fig:fig2ampdamppartial}
\end{figure}

\section{Closing Remarks}
In this work, we have presented a security proof of the HD-3-State-BB84 protocol and showed that, for high enough dimension, our work provides higher noise tolerances in the case when only one monitoring basis is used as compared to prior state of the art work for this protocol. The key advantage of our analysis is that it avoids computational limitations and provides an analytical expression for key rates in arbitrary dimensions. Our method also clearly demonstrates that indeed, using higher dimensional systems leads to an increment in noise tolerances even when Alice is limited in her ability to send the monitoring basis states. This conclusion could not be made in \cite{islam2018securing} for this three state protocol.  An interesting future work could be coming up with an analytical proof that increasing the number of monitoring basis states from Alice's end does improve the noise tolerance.

Another interesting line of investigation would be to analyze more practical channels, including lossy channels.  We assumed no loss in our channel model.  However, it is known that when loss is allowed, certain attacks such as unambiguous state discrimination attacks \cite{tamaki2003security,ko2018advanced,duvsek2000unambiguous,brandt2005unambiguous} become available to an adversary.  Since only $D+1$ states are transmitted from two bases (as opposed to all $2D$ states), this protocol may be particularly susceptible.  We leave this analysis as interesting, and important, future work.  However, we feel our proof methods may be suitable to tackle lossy conditions, with suitable extensions.

We have also presented a new continuity bound for conditional quantum entropies in this work, for certain types of cq-states. We have shown that, although limited in scope at this point, this new bound provides a noticeable advantage in noise tolerance in our analytical method, compared to Winter's continuity bound \cite{winter2016tight} (though, we stress, only for a certain type of state as our lemma is more restricted than Winter's bound, the latter of which can be applied to any state) and provides further support to Wilde's conjecture \cite{wilde2020optimal}. Perhaps more importantly, the technique we have used to prove our bound may find use in proving the conjecture itself, or some weaker version of it in arbitrary dimensions, as currently there are no known techniques to prove it. Our new continuity bound only applies to dimension two and in a limited noise scenario at this point and trying to extend it to higher dimensions is another interesting line of work for the future.

\bibliography{hdbb84-fewer-states}
\bibliographystyle{unsrt}

\section*{Appendix}
\subsection{High-dimensional entropic uncertainty relation}
Here we present the entropic uncertainty relation from \cite{krawec2020new} that we have used in our analysis of $\modepartial$ of protocol \eqref{prot:hd-3-State-BB84} where a two element POVM $\Lambda = \{\Lambda_0, \Lambda_1\}$ is used with $\Lambda_0 = \kb{x_0}{x_0}$ and $\Lambda_1 = I-\Lambda_0$. To state this relation, we need to describe an experiment \expmt. First, we consider a quantum state of the form $\rho_{TAE} = \sum_t p_t \kb{t}{t} \otimes \rho_{AE}^t$. Here the sum is over all subsets of $t$ of size $m$ and the Hilbert space on register $A$ has the form $\mathcal{H}_d^{\otimes(m + n)}$ with known $m, n, d$. In this experiment, we first measure the subset register $T$ to obtain an outcome $t$ and a post-measurement state $\rho_{AE}^t$. In this post-measurement state, indexed by $t$, we conduct our second measurement, where we measure the $d$-dimensional subspaces of register $A$ in POVM $\Lambda$ to obtain another outcome $q \in \{0, 1\}^m$. Tracing out the $T$ register and the measured portion of register $A$, leaves us with $n$ qudits of $d$-dimensional subspaces of $A$ unmeasured. We call it the post-measurement state $\rho(t, q)$. So at the end of the experiment $\expmt$, we obtain $(t, q, \rho(t, q)) \leftarrow \expmt(\rho_{TAE}, \Lambda)$. 
\newline\newline
\noindent The relation then bounds the min-entropy of the unmeasured portion of $A$ conditioned on register $E$, if $A$ were to be measured in an alternate basis $\Z$, based on the experiment's outcome $q$ and the basis choice $\Z$. If the Hamming weight of $q$ is small, we can argue that the conditional entropy of register $A$ given access to register $E$ should be high.

\begin{theorem}
	(From \cite{krawec2020new}). Let us consider an arbitrary quantum state $\rho_{AE}$ residing on Hilbert space $\mathcal{H_A} \otimes \mathcal{H_E}$, where the subspace of the subsystem $A$ has the structure $\mathcal{H_A} \cong \mathcal{H_D^{\otimes(\text{n + m})}}$, such that $D \ge 2$, and $m < n$. Let us also consider parameters $\epsilon > 0$, $0 < \beta < 1/2$ and two orthonormal bases of $\mathcal{H}_D$, namely, $Z = \{\ket{z_i}\}_{i = 0}^{D - 1}$ and $X = \{\ket{x_i}\}_{i = 0}^{D - 1}$, and a two outcome POVM $\Lambda$ with elements $\{\Lambda_0 = \kb{x_0}{x_0}, \Lambda_1 = \mathbb{I}   - \kb{x_0}{x_0} \}$. Now, we consider subsets $t \subset \{1,2,\cdots, n+m\}$ of size $m$ and let $T = {n + m \choose m}$. If we perform experiment {\normalfont $\expmt$}  on $\rho_{AE}$ and obtain:
	\begin{align*}
		(t, q, \rho(t, q)) \leftarrow  {\normalfont \expmt} \left(\frac{1}{T}\sum_t \kb{t}{t} \otimes \rho_{AE}, \Lambda \right),
	\end{align*}
	Then it holds that:	\label{thm:hd-ent-unc-krawec}
	\begin{align*}
		Pr\left(H_\infty^{\epsilon'}(Z | E)_{\rho(t, q)} + \frac{n \bar{H}_D(w(q) + \delta)}{\log_D 2} \ge n \gamma\right) \ge 1 - \epsilon''. 
	\end{align*}
	This probability is over the measurement outcomes $q$ and all choices of the subset $t$. Furthermore:
	\begin{align*}
		\gamma = -\log_2 max_{a, b \in \mathcal{A}_d} |\braket{z_a | x_b}|^2,
	\end{align*}
	and $\epsilon' = 4 \epsilon + 2 \epsilon^\beta$, $\epsilon'' = 2 \epsilon^{1 - 2\beta}$, and finally:
	\begin{align}
		\delta = \sqrt{\frac{(m + n + 2) \ln(2\epsilon^2)}{m(m + n)}}. 
	\end{align}
      \end{theorem}
      The above of course implies the following corollary which we use in our analysis:
      \begin{corollary}\label{cor:eu-new}
        Let $\rho_{AE}$ be a quantum state where the $A$ register consists of a single $D$-dimensional qudit.  Let $A_Z$ be the random variable resulting from measuring the $A$ register in the $\Z$ basis, and let $Q_X$ be the probability of observing outcome $\Lambda_1$ should the $A$ register be measured using POVM $\Lambda$ (described in the text above).  Then it holds that:
        \begin{equation}
          H(A_Z|E) + \frac{H_D(Q_X)}{\log_D2} \ge \gamma.
        \end{equation}
      \end{corollary}
      \begin{proof}
        This follows immediately from the quantum asymptotic equipartition property of min entropy \cite{tomamichel2009fully} and the law of large numbers.
      \end{proof}
\subsection{New continuity bound for conditional quantum entropies}
Continuity is a basic property of the von Neumann entropy of a quantum system $\rho$. Fannes \cite{fannes1973continuity} discovered an inequality that tells us how much the entropy would change if $\rho$ is changed by a small amount to create $\rho'$, based on the trace distance of the old state $\rho$ and the new state $\rho'$. This was later improved by Audenaert \cite{audenaert2007sharp} to the following, for density operators $\rho$ and $\sigma$ with a trace distance of $\epsilon$:
\begin{align*}
	|H(\rho) - H(\sigma)| \le \epsilon \log(D) + h(\epsilon).
\end{align*}
Along the same line, in \cite{winter2016tight}, Winter has presented a number of wonderful continuity bounds for conditional quantum entropies, including bounds for cq-states. Then, in \cite{wilde2020optimal}, Wilde proved the following proposition for such cq-states(using our notations).

\begin{proposition}
	\label{prop:wilde}
	\it{The following inequality holds for $\epsilon \in (0, 1 - 1/D_E]:$}
	\begin{align}
		|H(E|B^Z)_\rho - H(E|B^Z)_\sigma| \le \epsilon \log(D_E - 1) + h(\epsilon), \label{eq:wildeproved}
	\end{align}
	where $D_E$ is the dimension of system E, the states $\rho_{B^ZE}$ and $\sigma_{B^ZE}$ are the following finite-dimensional classical-quantum states:
	\begin{align*}
		\rho_{B^ZE} = \sum_{b} r(b) \kb{b}{b}_B \otimes \rho_E^b, \\
		\sigma_{B^ZE} = \sum_{b} s(b) \kb{b}{b}_B \otimes \sigma_E^b,
	\end{align*}
	r(x) and s(x) are probability distributions, $\{\rho_E^b\}_b$ and $\{\sigma_E^b\}_b$ are sets of states, the conditional entropy is defined in terms of the von Neumann entropy as $H(E | B^Z)_\rho := \sum_b r(b) H(\rho_E^b)$ and
	\begin{align*}
		\epsilon \ge \half \|\rho_{B^ZE} - \sigma_{B^ZE}\|_1.
	\end{align*}
	Also, there exists a pair of classical-quantum states saturating the bound for every value of $D_E$ and $\epsilon \in (0, 1 - 1/D_E]$.
\end{proposition}
Note that in proposition (\ref{prop:wilde}) mentioned above, sub-system $B^Z$ is classical, and sub-system $E$ is quantum. Then he asks the following question. If we flip the conditioning states in equation \eqref{eq:wildeproved}, i.e., from a quantum system conditioning on a classical system, to classical system conditioning on a quantum one, does the following hold?
\begin{align}
	|H(B^Z|E)_\rho - H(B^Z|E)_\sigma| \stackrel{?}{\le} \epsilon \log(D_B - 1) + h(\epsilon). \label{eq:wildeconjectured}
\end{align}
It is noteworthy that, Winter's bound in inequality \eqref{eq:wintercontinuityExtra}, in this same scenario gives a much weaker upper bound for the inequality above in \eqref{eq:wildeconjectured}. In the following sections, we show that a slightly weaker form of this conjecture holds in a limited scenario ($D = 2, 0 \le q \le .1464$).

\subsection{Continuity Bound in $D = 2$ and $0 \le q \le .1464$:}
The cq-states that we are considering here, are mentioned in equations  \eqref{eq:rhobefirst} and  \eqref{eq:sigbefirst}. In the rest of the discussions in this section, we are always considering $D = 2, 0 \le q \le .1464$, and the orthogonality assumption that $\bk{\eba}{\ebap} = 0$ if $a \ne a'$ (note we make no assumption on $\braket{e_b^a|e_{b'}^{a'}}$ for $b' \ne b$). We make a small comment about the higher dimensional case at the end of the section. Here we present the proof of lemma \eqref{lem:ourlemma} which claims that in our limited case, $\left|H(B^Z|E)_\rho - H(B^Z|E)_\sigma \right|\le h(1 - q - \mq)$.
\begin{lemma}
	\label{lem:ourlemma}
	For dimension of the classical sub-system $B$ denoted as $D_B = 2$, and with the assumption that $\bk{\eba}{\ebap} = 0$ if $a \ne a'$, in a depolarizing channel with noise parameter $0 \le q \le .1464$, the following holds for states $\rho_{B^ZE}$ and $\sigma_{B^ZE}$ (as defined in Equations \eqref{eq:rhobefirst} and \eqref{eq:sigbefirst} respectively):
	\begin{align*}
		\left|H(B^Z|E)_\rho - H(B^Z|E)_\sigma \right|\le h(1 - q -\mq).
	\end{align*}
\end{lemma}

\begin{proof}
	
	First, let's see that, by the definition of conditional entropies:
	\begin{align}
		&\;\;\;\;\left|H\left(B^Z|E\right)_\rho - H\left(B^Z|E\right)_\sigma \right| \nonumber \\
		&= \left|H\left(B^ZE\right)_\rho - H\left(E\right)_\rho - H\left(B^Z
		E\right)_\sigma + H\left(E\right)_\sigma \right| \label{eq:prflempart1}
	\end{align}
	Then, we remember that the joint entropies of the cq-states $\rho_{B^ZE}$ and $\sigma_{B^ZE}$  are \cite{wilde2011classical}:
	\begin{align}
		H(\rho_{B^ZE}) &= H\left(B^Z\right)_\rho + \onebyd \sum_{b} H\left(\sum_a \kb{\eba}{\eba}\right)_\rho, \label{eq:prflempart2}\\
		H(\sigma_{B^ZE}) &= H\left(B^Z\right)_\sigma + \onebyd \sum_{b} H\left(\sum_{a, a'} \kb{\eba}{\ebap}\right)_\sigma. \label{eq:prflempart3}
	\end{align}
	Here, note that $H(B^Z)_\rho$ and $H(B^Z)_\sigma$ cancel out in equation \eqref{eq:prflempart1}, as they are equal to one. This is not difficult to show as the probability of Bob's measuring any particular $b$, $b = 0$ for example (remembering that $D = 2$), in the depolarizing channel is: $p(0) = \sum_a p(0 | a) p(a) = \half(1 - q + q) = \half$, the same is true for $p(1)$, hence $H(B^Z)_\rho = H(B^Z)_\sigma = 1$. Furthermore, in equation \eqref{eq:prflempart2}, note that, for $b = 0$ and $b = 1$:
	\begin{align*}
		H\left(\sum_a \kb{e_0^a}{e_0^a} \right)_\rho = H\left(\sum_a \kb{e_1^a}{e_1^a} \right)_\rho = h(q),
	\end{align*}
	which can be seen from lemma \eqref{lem:consthq}.
	Additionally, for the state that arises for each particular $b$ in $\sigma_{B^ZE}$ in equation \eqref{eq:prflempart3}, denoted as $\tau \coloneqq \sum_{a, a'} \kb{\eba}{\ebap}_\sigma$, we see that its entropy is zero as explained in lemma \eqref{lem:rhosigmabinary}. So the term $H(B^ZE)_\sigma$ disappears in equation \eqref{eq:prflempart1}. Now, to calculate $H(E)_\rho$ and $H(E)_\sigma$ in equation \eqref{eq:prflempart1}, we need Eve's quantum states  $\rho_E$ and $\sigma_E$, which is found by tracing out Bob from $\rho_{B^ZE}$ and $\sigma_{B^ZE}$ respectively, i.e.,
	\begin{align}
		\rho_E = \Tr_B(\rho_{B^ZE}) &= \onebyd \sum_{b, a} \kb{\eba}{\eba},  \label{eq:defrhoe}\\
		\sigma_E = \Tr_B(\sigma_{B^ZE}) &= \onebyd \sum_{b, a, a'} \kb{\eba}{\ebap} \nonumber\\
		&= \onebyd \sum_{b, a} \kb{\eba}{\eba}  + \onebyd \sum_{b, a \ne a'} \kb{\eba}{\ebap}_\sigma \nonumber\\
		&=\rho_E + \Delta_E \label{eq:defsigmae},
	\end{align}
	where, $\Delta_E := \onebyd \sum_{b, a \ne a'} \kb{\eba}{\ebap}_\sigma$. The fact that we can write $\sigma_E$ as the sum of $\rho_E$ and some small ``noise'' $\Delta_E$ is convenient for our proof, as we shall see later. With the simplifications in equations \eqref{eq:prflempart2} and \eqref{eq:prflempart3}, and from the definitions of $\rho_E$ and $\sigma_E$, we can say that equation \eqref{eq:prflempart1} implies the following:
	\begin{align}
		&\;\;\;\; \left|H\left(B^Z|E\right)_\rho - H\left(B^Z|E\right)_\sigma \right| \notag \\
		&= \left | h(q) - H(E)_\rho + H(E)_{\rho + \Delta} \right| \label{eq:abscondent}.
	\end{align}
	Note that in the equation above, in the last term, we are not using $\sigma$ as the subscript, rather, using $\rho + \Delta $ to emphasize the fact that $\sigma_E = \rho_E + \Delta_E$.
	Let's define a set of vectors, which contains all possible vectors of eigenvalues that $\rho_E$ can have, which are prescribed by Horn's theorem \cite{horn1962eigenvalues}, as $\Gamma$. We describe Horn's theorem in some detail later. For notational convenience, we switch notations slightly and represent the entropy of a matrix by its eigenvalues rather than the matrix itself. So from now on, the entropy of the matrix $\rho_E$, while denoted by $H(E)_\rho$ so far, shall be denoted by $H(\vgm)$, with the understanding that $\vgm \in \Gamma$ is the vector of eigenvalues of $\rho_E$.
	
	Now we restate lemma \eqref{lem:ourlemma} in terms of these new notations. For each individual vector $\vgm$ prescribed by Horn's theorem, our claim in lemma \eqref{lem:ourlemma} then, is the following:
	\begin{align}
		\underset{\vgm \in \Gamma, v \in V}{\max} \big|h(q) -H(\vgm) + h(v) \big| \le h(1 - q - \sqrt{q (1 - q)}), \label{eq:mainclaim}
	\end{align}
	where $V \coloneqq \{v \in [0, 1] \mid  \gamma_1 - \sqrt{q (1 - q)} \le v \le \gamma_1 + \sqrt{q (1 - q)} \}$, and $\gamma_{1}$ is the first element of the vector of eigenvalues $\vgm$. The reason for defining this set $v$ in this range is explained in lemma \eqref{lem:rhosigmabinary}.
	We can make the following observations about the behavior of this maximization problem in equation \eqref{eq:mainclaim}. For any \textit{fixed} pair $\vgm \in \Gamma$ and $v \in V$, there are two possible cases that can occur. Either $H(\vgm) \ge h(q) + h(v)$ or $H(\vgm) \le h(q) + h(v)$, which we call the ``maxmin'' case and the ``minmax'' case, respectively. In the maxmin case, it is clear that, if we consider some other vector $\vgmmax$ such that $H\br{\vgmmax} \ge H(\vgm)$ and consider some $\vmin$ such that $h(\vmin) \le h(v)$, then we can get a larger absolute value in the left hand side of equation \eqref{eq:mainclaim}, compared to the one that we get for $\vgm$ and $v$. Similarly, in the second case where $H(\vgm) < h(q) + h(v)$, if we consider some $\vgmmin$ and $\vmax$ such that, $H(\vgmmin) \le H(\vgm)$ and $h(\vmax) \ge h(v)$, we can get an upper bound for the left hand side of equation \eqref{eq:mainclaim}. Now, our job is to create such vectors $\vgmmax, \vgmmin, \vmax, \vmin$, that provide upper and lower bounds for $H(\vgm)$ and $h(v)$ and show that these upper and lower bounds still produces a smaller absolute value than the right-hand side of inequality \eqref{eq:mainclaim}, thus proving lemma \eqref{lem:ourlemma}.
	As we shall see in the next sections, given any pair $\vgm$ and $v$, we can find $\vgmmax, \vgmmin, \vmax, \vmin$ by considering $\gmo$ of $\vgm$ only.
	Note that, we no longer have to deal with the absolute value function. Now proving lemma \eqref{lem:ourlemma} reduces to proving that, given a noise parameter $0 \le q \le .1464$, considering $\omqbtwo \le \gmo \le 1 - q$, for any pair of $\vgm$ and $v$ as defined above, the following holds:
	\begin{align*}
		\big|h(q) -H(\vgm) + h(v) \big| &\le  \max \left(H\br{\vgmmax} - h(q)- h(\vmin), h(\vmax) + h(q) - H\br{\vgmmin} \right) \\
		& \le h(1 - q - \mq).
	\end{align*}
	The above-mentioned range for $\gmo$ can be seen in lemma \eqref{lem:getlowupbound}. We now analyze both maxmin and minmax cases and show that they are upper bounded by $h(1 - q - \mq)$. We also note that, in the following analysis we are considering $q > 0$, as for $q = 0$, lemma \eqref{lem:ourlemma} holds trivially.  
	\subsection{Analysis of the Maxmin case:}
	First let's note that if $H(\vgm) \ge h(q) + h(v)$ for an arbitrary pair $\vgm$ and $v$, then clearly, it is the case that $|h(q) - H(\vgm) + h(v)| \le H\br{\vgmmax} - h(q)- h(\vmin)$. We can divide the right-hand side of this inequality into three different parts and argue that each of these parts is upper bounded by $h(1 - q - \mq)$. Because we can see that for a pair $\vgm$ and $v$, $\big|h(q) -H(\vgm) + h(v) \big|$ is upper bounded by a piecewise function:
	\begin{align}
		&\big|h(q) -H(\vgm) + h(v) \big| \nonumber \\
		& \le
		\begin{cases}
			H\br{\vgmmax} - h(q) - h\left(\gamma_1 - \mq \right), & \text{for } \frac{1 - q}{2} \le \gmo < \half \\
			H\br{\vgmmax} - h(q) - h \br{\gamma_1 + \mq}, & \text{for }  \half \le \gmo < 1 - \mq\\
			H\br{\vgmmax} - h(q), & \text{for } \gmo \ge 1 - \mq.
		\end{cases}\label{eq:maxminpiece}
	\end{align}
	Let's describe how we can see this. The first non-constant quantity in the equation, $H(\vgm)$, is always maximized for the vector $\br{\frac{1 - q}{2}, \frac{1 - q}{2}, \qbtwo, \qbtwo}$, which can be seen in lemma \eqref{lem:maxentvector}. This is true for any $\vgm$ and not just based on the particular $\gmo$ we are considering in this case. So, instead of going over all values of $v$ depending on $\gmo$ in the range $\gmo - \mq \le v \le \gmo + \mq$ and finding the minimum for $h(v)$, we can directly pick $h(\vmin)$ based on $\gamma_1$ by the following observation. We notice that, in the range $\frac{1 - q}{2} \le \gamma_1 < \half$, $h(\gmo - \mq) \le h(\gmo + \mq)$, because $\gmo - \mq$ is farther from $\half$, compared to $\gmo+ \mq$, as $\gmo < \half$. Consequently, $h(\gmo - \mq)$ would be the minimum for all $\gmo$ in this range. Furthermore, as $\gmo$ increases from the starting value of $\frac{1 - q}{2}$, $h(\gmo - \mq)$ increases with it in our considered range of $q$, because they are on the left side of the bell-shaped binary entropy function.
	Now, if we want to maximize the first case in Equation \eqref{eq:maxminpiece}, we may like to fix $h(v)$ to be fixed to the smallest possible value in this range of $\gmo$, which occurs for the first value of $\gmo$, $h(\frac{1 - q}{2} - \mq)$ and this is what we subtract here.
	So for the first case when $\frac{1 - q}{2} \le \gmo < \half$, we see that:
	\begin{align*}
		H\br{\vgmmax} - h(q) - h(\vmin) \le H\left(\frac{1 - q}{2}, \frac{1 - q}{2}, \qbtwo, \qbtwo\right) - h(q) - h\left(\frac{1 - q}{2} - \mq\right).
	\end{align*}
	We can easily show that, for $\half \le \gmo < 1 - \mq$,
	\begin{align*}
		H\br{\vgmmax} - h(q) - h(\vmin) \le H \br{1 - \mq, \mq - q, \qbtwo, \qbtwo} - h(q).
	\end{align*}
	The third case of inequality \eqref{eq:maxminpiece} happens for $1 - \mq \le \gmo \le 1 - q$. Again we see that in this range, because $ \gmo + \mq \ge 1$, $h(\vmin) = 0$.
	So in this case too, because $H(\vgm)$ is decreasing, we set it to the first value of $\gmo$ to get the maximum for $H(\cdot)$ in this range of $\gmo$, and see that:
	\begin{align*}
		H\br{\vgmmax} - h(q) - h(\vmin) &\le H\left(1 - \mq, \mq - q, \qbtwo, \qbtwo \right) - h(q).
	\end{align*}
	Finally, aggregating all three cases, for arbitrary $\vgm$, $v$ and when $H(\vgm) \ge h(q) + h(v)$:
	\begin{align*}
		&\big|h(q) -H(\vgm) + h(v) \big| \\
		& \le
		\begin{cases}
			H\left(\frac{1 - q}{2}, \frac{1 - q}{2}, \qbtwo, \qbtwo\right) - h(q) - h\left(\frac{1 - q}{2} - \mq\right), & \text{for } \frac{1 - q}{2} \le \gmo < \half \\
			H\br{1 - \mq, \mq - q, \qbtwo, \qbtwo} - h(q), & \text{for }  \half \le \gmo \le 1 - q.
		\end{cases}
	\end{align*}
	We can also easily show that the following two inequalities hold for $0 \le q \le .1464$:
	\small
	\begin{align}
		H\br{1 - \mq, \mq - q, \qbtwo, \qbtwo} - h(q) &\ge H\left(\frac{1 - q}{2}, \frac{1 - q}{2}, \qbtwo, \qbtwo\right) - h(q) - h\left(\frac{1 - q}{2} - \mq\right), \nonumber\\
		h\br{1 - q- \mq} &\ge H\br{1 - \mq, \mq - q, \qbtwo, \qbtwo} - h(q), \label{eq:maxminclmgrtsec}
	\end{align}
	\normalsize
	thus establishing inequality \eqref{eq:mainclaim} for the maxmin case. Now, we focus on the minmax case.

	\subsection{Analysis of the Minmax case:}
	Here, we consider the case when it is true that $H(\vgm) \le h(v) + h(q)$, given a pair $\vgm, v$. Then, for this pair:
	\begin{align*}
		h(q) + h(v) - H(\vgm) \le h(q) + h(\vmax) - H\br{\vgmmin}.  
	\end{align*}
	We can pick $v$ in the range $\gmo - \mq \le v \le \gmo + \mq$ that maximizes $h(v)$, which we denote as $h(\vmax)$ for different ranges of $\gamma_1$, based on only one selection criteria. If $\gamma_1 - \mq \ge .5$, $h(\vmax) = h(\gamma_1 - \mq)$, otherwise $.5$ is included in the range of possible $v$ values and we get $h(\vmax) = 1$. Note that one possible case, when $\gamma_1 + \mq < .5$, occurs only for $q \ge .8536$ and is outside the scope of our considered range of $q$. Now we say that for this minmax case, the following holds:
	\begin{align*}
		&\big|h(q) -H(\vgm) + h(v) \big| \\
		& \le
		\begin{cases}
			1 + h(q) - H\br{\vgmmin}, & \text{for } \frac{1 - q}{2} \le \gamma_1 < \half + \mq \\
			h(\gamma_1 - \mq) + h(q) - H\br{\vgmmin}, & \text{for } \half + \mq \le \gmo \le 1 - q
		\end{cases}
	\end{align*}
	In the first case above, we consider the range $\frac{1 - q}{2} \le \gamma_1 < \half + \mq$. As we have argued earlier, in this range of $\gmo$ and $\gmo - \mq \le v \le \gmo + \mq$, $\half$ is included in the range of possible values of $v$, so $h(\vmax) = 1$. Because $1$ and $h(q)$ is a constant given $q$, the maximum of the first case would occur when we subtract the least $H\br{\vgmmin}$, which we evaluate for the $\vgmmin$ that we construct with our strategy described in lemma \eqref{lem:argminmax}. We find that $\vgmmin = \left( \max(\gmo, 1 - \gmo), \min(\gmo, 1 - \gmo) , 0, 0 \right)$, and based on the nature of the binary entropy function, for increasing value of $\gamma_1$, $H\br{\vgmmin}$ would decrease.  
	The second case happens for the minimum of $H\br{\vgmmin}$ is of course when $\gamma_1$ can reach its largest possible value $1 - q$, and we get the vector $\vgmmin = (1 - q, q, 0, 0)$, which can be seen in lemma \eqref{lem:argminmax}. In which case, $h(1 - q)$ equals $h(q)$ and they cancel out from the equation, leaving us with only $h(\gmo - \mq)$. But in this range of $\gmo$ for all $q$ in our consideration, $h(\gmo - \mq)$ is an increasing function and attains its maximum for $\gmo = 1 - q$.
	So we say the following:
	\begin{align*}
		&\big|h(q) -H(\vgm) + h(v) \big| \\
		& \le
		\begin{cases}
			h(1 - q- \mq), & \text{for } \gamma_1 - \mq \ge \half \\
			1 + h(q) - h\left(\half + \mq \right), & \text{for } \frac{1 - q}{2} \le \gamma_1 < \half + \mq.  
		\end{cases}
	\end{align*}
	At this point, we have to show that, our claim in inequality \eqref{eq:mainclaim} is valid for all $q$ in our considered range. More precisely, we show the following inequality is valid:
	\begin{align}
		h(1 - q- \mq) \ge 1 + h(q) - h\br{\half + \mq}. \label{eq:minmaxsecgrtone}
	\end{align}
	Inequality \eqref{eq:minmaxsecgrtone} is true for $0 <  q \le .146447$ which can be seen by graphs or using numerical methods.  So, considering inequalities \eqref{eq:maxminclmgrtsec} and \eqref{eq:minmaxsecgrtone}, we can say that inequality \eqref{eq:mainclaim}, holds for $0 \le q \le .1464$ in all possible cases, which proves lemma \eqref{lem:ourlemma}.

\end{proof}  

\begin{lemma}
	\label{lem:consthq}
	For depolarizing channel parameter $q$, with the assumption that $\bk{\eba}{\ebap} = 0 $ if $a \ne a'$, the following holds for $D_B = 2$:
	\begin{align*}
		H\left(\sum_a \kb{e_0^a}{e_0^a}\right) = H\left(\sum_a \kb{e_1^a}{e_1^a}\right) = h(q).
	\end{align*}
\end{lemma}
\begin{proof}
	In a depolarizing channel scenario, the eigenvalues for a specific $b$, $b = 0$ for example, the matrix $\rho_E^0 \coloneqq \sum_a \kb{e^a_{0}}{e^a_{0}}$ would have the same eigenvalues as another matrix for $b = 1$, $\rho_E^1 \coloneqq \sum_a \kb{e^a_{1}}{e^a_{1}}$. Because for a given $b$, our assumption that $\bk{\eba}{\ebap} = 0 $ if $a \ne a'$, gives rise to a basis. Moreover, the depolarizing channel scenario lets us determine the exact eigenvalues of $\rho_{b}$, for both $b = 0$ and $b = 1$. These are, in a non-increasing order, $\{1 - q, q, 0, 0\}$ for a channel parameter $q$, because $\bk{\ezz}{\ezz} = 1 - q$ and $\bk{\ezo}{\ezo} = q$. Similarly, it is the case that, $\bk{\ezo}{\ezo} = q$ and $\bk{\eoo}{\eoo} = 1 - q$. Note that the dimension of each of the states in Eve's memory $\ket{e_i^j}$ for dimension $D_B = 2$, can be taken to be four \cite{qkd-survey-pirandola}.
	
\end{proof}

\begin{lemma}
	\label{lem:argmaxmin}
	Given the noise parameter $q$, and considering the set $\Gamma$ as the set of vectors of all possible eigenvalues for $\rho_E$ and for an arbitrary $\vgm$:
	\begin{align*}
		\vgmmax &= \left( \gmo, 1 - \gmo - q, \qbtwo, \qbtwo \right), \\
		\vgmmin &= \left( \max(\gmo, 1 - \gmo), \min(\gmo, 1 - \gmo) , 0, 0 \right).
	\end{align*}
\end{lemma}
\begin{proof}
	
	Notice that, for an arbitrary but fixed $\vgm$, if we fix the first element $\gmo$, fix the third and fourth elements to be $\qbtwo$ and place the remaining weights in the second place, we get a vector, which we call $\vgmmax$ such that, $H\br{\vgmmax} \ge H(\vgm)$. It is possible to fix the third and fourth elements to $\qbtwo$ because according to Horn's theorem,
	\begin{align*}
		\gamma_3 + \gamma_4 \le \alpha_2 + \alpha_3 + \beta_2 + \beta_3,
	\end{align*}
	which in our case, translates to the following:
	\begin{align*}
		\gamma_3 + \gamma_4 \le \qbtwo + \qbtwo = q,
	\end{align*}
	(see lemma \eqref{lem:getlowupbound}). Horn's theorem also tells us that,
	\begin{align}
		&\gamma_1 + \gamma_2 \ge \alpha_1 + \alpha_4 + \beta_1 + \beta_4 = \omqbtwo + \omqbtwo = 1 - q, \label{eq:horn12}\\
		&\gamma_1 + \gamma_2 + \gamma_3 \ge \alpha_1 + \alpha_3 + \alpha_4 + \beta_1 + \beta_2 + \beta_4 = 1 - q + \frac{q}{2}. \label{eq:horngm123}
	\end{align}
	It is easy to see that, for this vector $\vgmmax$, $H\br{\vgmmax} \ge H(\vgm)$, because any arbitrary $\vgm$ would majorize $\vgmmax$ \cite{nielsen2001majorization}, \cite{sagawa2020entropy}. Notice that the first elements of an arbitrary $\vgmmax$ and $\vgm$ are the same. Moreover, the sum of the first two elements of $\vgm$ is at least $1 - q$, and the sum of the first three is at least $1 - q + \qbtwo$ as can be seen from equations \eqref{eq:horn12} and \eqref{eq:horngm123}. This majorizes the first three elements of $\vgmmax$ that we are constructing. The fourth element at this point can't affect the whole majorization, which can be seen at a moment's thought.
	
	We can also find $\vgmmin$ by constructing a vector that only has two non-zero values, one of the elements is $\gamma_1$ and the other one is $1 - \gamma_1$ and considering the larger one as the first element. In this way, we get $\vgmmin = \left( \max(\gmo, 1 - \gmo), \min(\gmo, 1 - \gmo) , 0, 0 \right)$. Obviously, this $\vgmmin$, would majorize any $\vgm$ and consequently for any arbitrary $\vgm$, $H(\vgm) \ge H\br{\vgmmin}$.
\end{proof}

\begin{lemma}
	\label{lem:maxentvector}
	Given the noise parameter $q$, for all $\omqbtwo \le \gmo \le 1 - q$, the maximum for $H\br{\vgmmax}$ is achieved when $\vgmmax = (\omqbtwo, \omqbtwo, \qbtwo, \qbtwo)$.
\end{lemma}
\begin{proof}
	In order to establish the fact that $\vgmmax = (\omqbtwo, \omqbtwo, \qbtwo, \qbtwo)$ for all possible $\gmo$, we remember the definition of $\hbzerho$. Clearly, $\hbzerho = H(B^ZE)_\rho - H(E)_\rho$ implies the following, considering $D_B = 2$:
	\begin{align*}
		H(E)_\rho &\le H(B^ZE)_\rho \\
		&= H\left(B^Z\right) + \onebyd \sum_{b} H\left(\sum_a \kb{\eba}{\eba}\right) \\
		&= 1 + \half \left(H\left(\rho_E^0\right) + H(\rho_E^1) \right) \\
		&= 1 + h(q).
	\end{align*}
	The last equality follows because of lemma \eqref{lem:consthq}. It is easy to see that, equality for this upper bound is achieved by one of the vector of eigenvalues of $\rho_E$, $\vgmmax = \left( \frac{1 - q}{2}, \frac{1 - q}{2}, \qbtwo, \qbtwo \right)$. Now we explain how this $\vgmmax$ can be derived. Notice that, both $\rho_E^0$ and $\rho_E^1$ have eigenvalues $\{\frac{1 - q}{2}, \frac{q}{2}, 0, 0\}$ as explained in lemma \eqref{lem:consthq}. Let's denote eigenvalues of $\rho_E^0$ as $\vec{\alpha}$ and the eigenvalues of $\rho_E^1$ as $\vec{\beta}$, then, $\vec{\alpha} = \vec{\beta} = \{\frac{1 - q}{2}, \frac{q}{2}, 0, 0\}$. We further denote the eigenvalues of the resulting matrix by $\vec{\gamma}$. Now, one of the inequalities from Horn's theorem says the following:
	\begin{align*}
		\gamma_3 + \gamma_4 \le \alpha_2 + \alpha_3 + \beta_2 + \beta_3.
	\end{align*}
	This inequality upper bounds the third and fourth element of the resulting matrix $\rho_E^0 + \rho_E^1$ by the second and third element of the summands $\rho_E^0$ and $\rho_E^1$. In our context, this means that $\gamma_3 + \gamma_4 \le \frac{q}{2} + \frac{q}{2} = q$. Keeping in mind that $\gamma_4 \le \frac{q}{2}$ from lemma \eqref{lem:getlowupbound} and remembering that uniform distribution maximizes entropy, we determine $\gamma_3 = \frac{q}{2}, \gamma_4 = \frac{q}{2}$. Then we are left with distributing $1 - q$ between $\gamma_1$ and $\gamma_2$. Again, we distribute this weight evenly and determine $\gamma_1 = \gamma_2 = \frac{1- q}{2}$, remembering that $\gamma_1 \ge \frac{1 - q}{2}$ from lemma \eqref{lem:getlowupbound}.
	
	
\end{proof}

\begin{lemma}
	\label{lem:argminmax}
	Given the noise parameter $q$, for all $\omqbtwo \le \gmo \le 1 - q$, the minimum for $H\br{\vgmmin}$ is achieved when $\vgmmin = (1 - q, q, 0, 0)$.
\end{lemma}
\begin{proof}
	Similar to the case of lemma \eqref{lem:argmaxmin}, we start with the entropy of matrix $\rho_E = \onebyd \sum_{b, a} \kb{\eba}{\eba}$ and keep in mind that $D_B = 2$.
	\begin{align*}
		H(E)_\rho &\ge H(B^ZE)_\rho - H(B^Z)_\rho \\
		&=  H\left(B^Z\right)_\rho + \onebyd \sum_{b} H\left(\sum_a \kb{\eba}{\eba}\right)_\rho - H\left(B^Z\right)_\rho \\
		&= \onebyd \sum_{b} H\left(\sum_a \kb{\eba}{\eba}\right)_\rho \\
		&= \half (H(\rho_0) + H(\rho_1)) \\
		&= h(q).
	\end{align*}
	The last equality follows because of lemma \eqref{lem:consthq}. To see how a vector with this minimum entropy can be achieved from Horn's theorem, let's remember an upper bound on the first element of the set of eigenvalues of the matrix $\rho_E$, denoted by $\vgm$, which is:
	\begin{align*}
		\gamma_1 \le \alpha_1 + \beta_1.
	\end{align*}
	In our case where $\vec{\alpha} = \{\frac{1 - q}{2}, \frac{q}{2}, 0, 0\}$ and $\vec{\beta} = \{\frac{1 - q}{2}, \frac{q}{2}, 0, 0\}$, this implies $\gamma_1 \le \frac{1 - q}{2} + \frac{1 - q}{2} = 1 - q$. We set $\gamma_1 = 1 - q$ and set the remaining probability weight on $\gamma_2$. Thus getting a vector $\vgmmin = \{1 - q, q, 0, 0\}$. It is easy to see that any other vector $\vgm'$ will be majorized by $\vgmmin$ and will have larger entropy according to the majorization theorem \cite{nielsen2001majorization}, \cite{sagawa2020entropy}.
	
\end{proof}

\begin{lemma}
	\label{lem:getlowupbound}
	For depolarizing channel parameter $q$, the matrix $\rho_E$ found in equation \eqref{eq:defrhoe}, can only have eigenvalues in the following range:
	\begin{align*}
		&\frac{1 - q}{2} \le \gamma_{1} \le 1 - q, \frac{q}{2} \le \gamma_2 \le \half,  \\
		& 0 \le \gamma_3 \le  q,  0 \le \gamma_4 \le \frac{q}{2}.
	\end{align*}
\end{lemma}

\begin{proof}
	For $D_B = 2$ we can see from equation \eqref{eq:defrhoe} that,
	
	\begin{align}
		\rho_E &= \half  \rho_E^0 + \half \rho_E^1,  \label{eq:defM}
	\end{align}
	where,
	\begin{align*}
		\rho_E^0 = \kb{e_0^0}{e_0^0} + \kb{e_0^1}{e_0^1}, \;\;\;\;\;\;\;\; \rho_E^1 = \kb{e_1^0}{e_1^0} + \kb{e_1^1}{e_1^1}.
	\end{align*}
	Let's remember the definition of $\rho_E$ from equation \eqref{eq:defM}, and apply the inequalities derived from Horn's theorem to bound its eigenvalues from the eigenvalues of $\rho_E^0$ and $\rho_E^1$. Note that the eigenvalues of $\half \rho_E^0$ are simply the same eigenvalues of $\rho_E^0$, scaled by half. So in fact, the two summands $\half \rho_E^0$ and $ \half \rho_E^1$, both have these eigenvalues: $\{\frac{1 - q}{2}, \qbtwo, 0, 0\}$ (see lemma \eqref{lem:consthq}), where the elements are arranged in a non-increasing order. We call these two vectors $\vec{\alpha}$ and $\vec{\beta}$ and the elements by their respective position index. Then let's define the vector of the resulting matrix $\rho_E$'s eigenvalues from these summands as:
	\begin{align*}
		\vgm = \{\gamma_1, \gamma_2, \gamma_3, \gamma_4\},
	\end{align*}
	where the elements are in non-increasing order. First, let's focus on the upper bound of $\gamma_1$. We can easily see that, because $\gamma_1 \le \alpha_1 + \beta_1$ by Horn's theorem,
	\begin{align}
		\gamma_1 \le \frac{1-q}{2} + \frac{1-q}{2} = 1 - q \label{eq:gam1upper}.
	\end{align}
	
	Furthermore, because $\gamma_1 \ge \alpha_1 + \beta_4$,
	\begin{align}
		\gamma_1 \ge \frac{1-q}{2} + 0 = \frac{1-q}{2} \label{eq:gam1lower}.
	\end{align}
	Similarly, from $\gamma_2 \le \alpha_1 + \beta_2$ and $\gamma_2 \ge \alpha_4 + \beta_2$,$\gamma_3 \le \alpha_2 + \beta_2$  and $\gamma_3 \ge \alpha_4 + \beta_3$, $\gamma_4 \le \alpha_2 + \beta_3$ and  $\gamma_4 \ge \alpha_4 + \beta_4$ ,
	\begin{align}
		\qbtwo \le \gamma_2 &\le \frac{1 - q}{2} + \qbtwo = .5 \label{eq:gam2lowup} \\
		0 \le \gamma_3 &\le \qbtwo + \qbtwo \le q\label{eq:gam3lowup} \\
		0 \le \gamma_4 &\le \qbtwo\label{eq:gam4lowup}.
	\end{align}
\end{proof}
\begin{lemma}
	\label{lem:del2upper}
	For the matrix $\Delta_E$ defined in equation \eqref{eq:abscondent}, the following holds for depolarizing noise parameter $q$:
	\begin{align*}
		\|\Delta_E\|_{op} \le \mq.
	\end{align*}
\end{lemma}
\begin{proof}
	\begin{align*}
		\|\Delta_E\|_{op} &= \|\onebyd \sum_{b, a \ne a'} \kb{\eba}{\ebap}\|_{op} \\
		&= \onebyd \|\sum_{b, a \ne a'} \kb{\eba}{\ebap}\|_{op} \\
		&\le \half \|\kb{\ezz}{\ezo} + \kb{\ezo}{\ezz}\|_{op} + \half \|\kb{\eoz}{\eoo} + \kb{\eoo}{\eoz}\|_{op}   \text{  [triangle inequality ]}\\  
		&= \half \sqrt{\gamma_{max}((\kb{\ezz}{\ezo} + \kb{\ezo}{\ezz})^\dagger (\kb{\ezz}{\ezo} + \kb{\ezo}{\ezz}))} \\&+ \half \sqrt{\gamma_{max}((\kb{\eoz}{\eoo} + \kb{\eoo}{\eoz})^\dagger (\kb{\eoz}{\eoo} + \kb{\eoo}{\eoz}))}\\  
		&= \half \sqrt{\gamma_{max}((\kb{\ezz}{\ezo} + \kb{\ezo}{\ezz})(\kb{\ezz}{\ezo} + \kb{\ezo}{\ezz}))} \\&+ \half \sqrt{\gamma_{max}((\kb{\eoz}{\eoo} + \kb{\eoo}{\eoz})(\kb{\eoz}{\eoo} + \kb{\eoo}{\eoz}))}\\  
		&= \half \sqrt{\gamma_{max}( \bk{\ezo}{\ezo} \kb{\ezz}{\ezz} + \bk{\ezz}{\ezz}\kb{\ezo}{\ezo})} \\&+ \half \sqrt{\gamma_{max}(\bk{\eoo}{\eoo} \kb{\eoz}{\eoz} + \bk{\eoz}{\eoz} \kb{\eoo}{\eoo})},
	\end{align*}
	where, in the fourth line, we have used the definition of the operator norm, which is the largest singular value of a matrix. Now, we remember that for the depolarizing channel that we are considering, $\bk{\ezz}{\ezz} = 1 - q$ and $\bk{\ezo}{\ezo} = q$, similarly, $\bk{\ezo}{\ezo} = q$ and $\bk{\eoo}{\eoo} = 1 - q$. Then we write $\ket{\ezz} = (1 - q) \ket{m}$ and $\ket{\ezo} = q\ket{m^\perp}$, for some orthonormal basis $ \{ \ket{m}, \ket{m^\perp}\} $. Similarly, we consider  $\ket{\eoz} = q\ket{n}$ and $\ket{\eoo} = (1 - q)\ket{n^\perp}$, for some orthonormal basis $ \{ \ket{n}, \ket{n^\perp}\} $. Then the inequality above can be written as:
	\begin{align*}
		\| \Delta_E \|_{op} &\le \half \sqrt{\gamma_{max}(q(1 - q)\kb{m}{m} + q(1 - q) \kb{m^\perp}{m^\perp})} \\
		&+ \half \sqrt{\gamma_{max}(q(1 - q)\kb{n}{n} + q(1 - q) \kb{n^\perp}{n^\perp})}\\
		&=\mq.
	\end{align*}
	
\end{proof}

\begin{lemma}
	\label{lem:rhosigmabinary}
	\[
	H(E)_{\rho + \Delta} \in \{h(v) | v \in V\},
	\]
	where $V \coloneqq \{v \in [0, 1] \mid  \gamma_1 - \sqrt{q (1 - q)} \le v \le \gamma_1 + \sqrt{q (1 - q)} \}$, and $\gamma_{1}$ is the first element of the vector $\vgm$, which is the vector of eigenvalues of $\rho_E$.
\end{lemma}
\begin{proof}
	Because $\rho_E$, $\Delta_E$ are hermitian, $\rho_E + \Delta_E$ is also hermitian. Now, according to Weyl's eigenvalue stability theorem \cite{veselic2007spectral, tao2011topics}, the eigenvalues of $\rho_E + \Delta_E$ would not vary too much from the eigenvalues $\rho_E$. More precisely, each pair of eigenvalues of $\rho_E$ and $\rho_E + \Delta_E$, after arranging them in non-decreasing order, is related by the following quantity:
	\begin{align}
		|\mu_k - \delta_k| \le \|\Delta_E\|_{op},
	\end{align}
	where, $\mu_k$ and $\delta_k$ are the eigenvalues of $\rho_E$ and $\rho_E + \Delta_E$ respectively. So, we need to find an upper bound on  $\|\Delta_E\|_{op}$ now. From lemma \eqref{lem:del2upper}, we know that $\|\Delta_E\|_{op} \le \sqrt{q(1 - q)}$. Let's argue further that $\rho_E + \Delta_E$ can have at most two eigenvalues for $D = 2$, because we can notice that, in this dimension:
	\begin{align*}
		\rho_E + \Delta_E &= \half \big( \kb{\ezz}{\ezz} + \kb{\ezz}{\ezo} + \kb{\ezo}{\ezz} + \kb{\ezo}{\ezo} \\
		&\;\;\;\;\;+ \kb{\eoz}{\eoz} + \kb{\eoz}{\eoo} + \kb{\eoo}{\eoz} + \kb{\eoo}{\eoo}\big).
	\end{align*}
	Let $\tau \coloneqq \kb{\ezz}{\ezz} + \kb{\ezz}{\ezo} + \kb{\ezo}{\ezz} + \kb{\ezo}{\ezo}$, i.e., the case when $b = 0$.
	Now, with our assumption that $\bk{\eba}{\ebap} = 0$ when $a \ne a'$, we see that $\tau \times \tau = \tau$ and it has unit trace as a valid quantum state. The same holds for $b = 1$. So, $\rho_E + \Delta_E$ then, is simply an equal mixture of two pure states, and pure states have unit ranks, consequently, $\rho_E + \Delta_E$ has rank 2, i.e., two eigenvalues. Hence, we calculate its von Neumann entropy with the binary entropy function in the range $V$.
\end{proof}

\subsection{Horn's Theorem}
Horn's conjecture \cite{horn1962eigenvalues, fulton2000eigenvalues, bhatia2001linear}, which is now a theorem thanks to Klyachko \cite{klyachko1998stable}, and Knutson and Tao \cite{knutson1999honeycomb}, enables one to bound the eigenvalues of some Hermitian matrix $C = A + B$, where $A$ and $B$ are Hermitian with known spectrum. We can denote the spectrums of $A$ and $B$ by $\alpha = \alpha_1 \ge \alpha_2 \ge ... \ge \alpha_n$, $\beta = \beta_1 \ge \beta_2 \ge ... \ge \beta_n$, and for $C$, by $\gamma = \gamma_1 \ge \gamma_2 \ge ... \ge \gamma_n$. This theorem then prescribes a set of triplets $(\alpha, \beta, \gamma)$ which can be the eigenvalues of $A, B$, and $C = A + B$ respectively. More specifically, it produces sets (following notations from \cite{fulton2000eigenvalues}) $T^n_r$ of triplets $(I, J, K)$ of subsets of $\{1, ... , n\}$, where the cardinalities of each $I, J, K$ are the same and is $r$. This set $T^n_r$ needs to be built recursively. First, one builds the following set:
\begin{align*}
	U_r^n = \left \{(I, J, K) \;\;|\;\; \sum_{i \in I} i + \sum_{j \in J} j = \sum_{k \in K} k + \frac{r(r + 1)}{2}\right \}.
\end{align*}
If, $r = 1$, then we are done, because in this case, $T^n_r = U_r^n$. This basic case of $r = 1$ is also known as Weyl's inequalities \cite{horn2012matrix}. Now, one builds the following set:
\begin{align}
	T^n_r &= \{ (I, J, K) \in U^n_r \;\;|\;\; \text{for all } p < r \text{ and all }(F, G, H) \text{ in } T^r_p \text{ such that }, \nonumber \\
	& \; \;\;\;\;\;\;\;\;\; \; \;\;\;\;\;\;\; \sum_{f \in F} i_f + \sum_{g \in G} j_g \le \sum_{h \in H} k_h + \frac{p(p + 1)}{2} \label{horn_trn} \}.
\end{align}
Now we state the theorem:
\paragraph{Horn's theorem: \cite{horn1962eigenvalues}}
A triple $(\alpha, \beta, \gamma)$ occurs as eigenvalues of Hermitian $n \times n$ matrices $A, B$ and $C$ with $C = A + B$ respectively, if and only if
\begin{align*}
	\sum \gamma_i = \sum \alpha_i + \sum \beta_i,
\end{align*}  and the inequalities
\begin{align*}
	\sum_{k \in K} \gamma_k \le \sum_{i \in I} \alpha_i + \sum_{j \in J} \beta_j,
\end{align*}
hold for every $(I, J, K)$ in $T^n_r$, for all $r < n$.
\newline \newline
Below, we list all the inequalities for upper and lower bounding eigenvalues of the sum matrix $C$, when $n = 4$, which is the case applicable to our lemma. We can easily get the lower bounds for the above inequalities also. Because, every triplet $(I, J, K)$ from \eqref{horn_trn} can be subtracted from the equality condition:
\begin{align*}
	\sum_{i = 1}^n \gamma_i = \sum_{i = 1}^n \alpha_i + \sum_{i = 1}^n \beta_i ,
\end{align*}
to get:
\begin{align*}
	\sum_{k \in K^c} \gamma_k \ge  \sum_{i \in I^c} \alpha_i + \sum_{j \in J^c} \beta_j.
\end{align*}
Here we list all the upper bounds of the eigenvalues of $C = A + B$ from Horn's theorem when they are all hermitian $4 \times 4$ matrices.
For $r = 1$ we have:
\begin{align*}
	\gamma_1\le\alpha_1+\beta_1
	\;\;\;\;\;\;\;\;\;
	\gamma_2\le\alpha_1+\beta_2
	\;\;\;\;\;\;\;\;\;
	\gamma_3\le\alpha_1+\beta_3
	\\
	\gamma_4\le\alpha_1+\beta_4
	\;\;\;\;\;\;\;\;\;
	\gamma_2\le\alpha_2+\beta_1
	\;\;\;\;\;\;\;\;\;
	\gamma_3\le\alpha_2+\beta_2
	\\
	\gamma_4\le\alpha_2+\beta_3
	\;\;\;\;\;\;\;\;\;
	\gamma_3\le\alpha_3+\beta_1
	\;\;\;\;\;\;\;\;\;
	\gamma_4\le\alpha_3+\beta_2
	\\
	\gamma_4\le\alpha_4+\beta_1
\end{align*}
For $r = 2$:
\begin{align*}
	\gamma_1+\gamma_2\le\alpha_1+\alpha_2+\beta_1+\beta_2
	\\
	\gamma_1+\gamma_3\le\alpha_1+\alpha_2+\beta_1+\beta_3
	\;\;\;\;\;\;\;\;\;
	\gamma_1+\gamma_4\le\alpha_1+\alpha_2+\beta_1+\beta_4
	\\
	\gamma_2+\gamma_3\le\alpha_1+\alpha_2+\beta_2+\beta_3
	\;\;\;\;\;\;\;\;\;
	\gamma_2+\gamma_4\le\alpha_1+\alpha_2+\beta_2+\beta_4
	\\
	\gamma_3+\gamma_4\le\alpha_1+\alpha_2+\beta_3+\beta_4
	\;\;\;\;\;\;\;\;\;
	\gamma_1+\gamma_3\le\alpha_1+\alpha_3+\beta_1+\beta_2
	\\
	\gamma_1+\gamma_4\le\alpha_1+\alpha_3+\beta_1+\beta_3
	\;\;\;\;\;\;\;\;\;
	\gamma_2+\gamma_3\le\alpha_1+\alpha_3+\beta_1+\beta_3
	\\
	\gamma_2+\gamma_4\le\alpha_1+\alpha_3+\beta_1+\beta_4
	\;\;\;\;\;\;\;\;\;
	\gamma_2+\gamma_4\le\alpha_1+\alpha_3+\beta_2+\beta_3
	\\
	\gamma_3+\gamma_4\le\alpha_1+\alpha_3+\beta_2+\beta_4
	\;\;\;\;\;\;\;\;\;
	\gamma_1+\gamma_4\le\alpha_1+\alpha_4+\beta_1+\beta_2
	\\
	\gamma_2+\gamma_4\le\alpha_1+\alpha_4+\beta_1+\beta_3
	\;\;\;\;\;\;\;\;\;
	\gamma_3+\gamma_4\le\alpha_1+\alpha_4+\beta_1+\beta_4
	\\
	\gamma_2+\gamma_3\le\alpha_2+\alpha_3+\beta_1+\beta_2
	\;\;\;\;\;\;\;\;\;
	\gamma_2+\gamma_4\le\alpha_2+\alpha_3+\beta_1+\beta_3
	\\
	\gamma_3+\gamma_4\le\alpha_2+\alpha_3+\beta_2+\beta_3
	\;\;\;\;\;\;\;\;\;
	\gamma_2+\gamma_4\le\alpha_2+\alpha_4+\beta_1+\beta_2
	\\
	\gamma_3+\gamma_4\le\alpha_2+\alpha_4+\beta_1+\beta_3
	\;\;\;\;\;\;\;\;\;
	\gamma_3+\gamma_4\le\alpha_3+\alpha_4+\beta_1+\beta_2
\end{align*}
For $r = 3$:
\begin{align*}
	\gamma_1+\gamma_2+\gamma_3\le\alpha_1+\alpha_2+\alpha_3+\beta_1+\beta_2+\beta_3
	\\
	\gamma_1+\gamma_2+\gamma_4\le\alpha_1+\alpha_2+\alpha_3+\beta_1+\beta_2+\beta_4
	\\
	\gamma_1+\gamma_3+\gamma_4\le\alpha_1+\alpha_2+\alpha_3+\beta_1+\beta_3+\beta_4
	\\
	\gamma_2+\gamma_3+\gamma_4\le\alpha_1+\alpha_2+\alpha_3+\beta_2+\beta_3+\beta_4
	\\
	\gamma_1+\gamma_2+\gamma_4\le\alpha_1+\alpha_2+\alpha_4+\beta_1+\beta_2+\beta_3
	\\
	\gamma_1+\gamma_3+\gamma_4\le\alpha_1+\alpha_2+\alpha_4+\beta_1+\beta_2+\beta_4
	\\
	\gamma_2+\gamma_3+\gamma_4\le\alpha_1+\alpha_2+\alpha_4+\beta_1+\beta_3+\beta_4
	\\
	\gamma_1+\gamma_3+\gamma_4\le\alpha_1+\alpha_3+\alpha_4+\beta_1+\beta_2+\beta_3
	\\
	\gamma_2+\gamma_3+\gamma_4\le\alpha_1+\alpha_3+\alpha_4+\beta_1+\beta_2+\beta_4
	\\
	\gamma_2+\gamma_3+\gamma_4\le\alpha_2+\alpha_3+\alpha_4+\beta_1+\beta_2+\beta_3
\end{align*}

Now we list the lower bounds for the eigenvalues of the resulting sum matrix $C = A + B$. For $r = 1$:

\begin{align*}
	\gamma_4\ge\alpha_4+\beta_4
	\;\;\;\;\;\;\;\;\;
	\gamma_3\ge\alpha_4+\beta_3
	\;\;\;\;\;\;\;\;\;
	\gamma_2\ge\alpha_4+\beta_2
	\\
	\gamma_1\ge\alpha_4+\beta_1
	\;\;\;\;\;\;\;\;\;
	\gamma_3\ge\alpha_3+\beta_4
	\;\;\;\;\;\;\;\;\;
	\gamma_2\ge\alpha_3+\beta_3
	\\
	\gamma_1\ge\alpha_3+\beta_2
	\;\;\;\;\;\;\;\;\;
	\gamma_2\ge\alpha_2+\beta_4
	\;\;\;\;\;\;\;\;\;
	\gamma_1\ge\alpha_2+\beta_3
	\\
	\gamma_1\ge\alpha_1+\beta_4
\end{align*}
For $r = 2$:
\begin{align*}
	\gamma_3+\gamma_4\ge\alpha_3+\alpha_4+\beta_3+\beta_4
	\\
	\gamma_2+\gamma_4\ge\alpha_3+\alpha_4+\beta_2+\beta_4
	\;\;\;\;\;\;\;\;\;
	\gamma_2+\gamma_3\ge\alpha_3+\alpha_4+\beta_2+\beta_3
	\\
	\gamma_1+\gamma_4\ge\alpha_3+\alpha_4+\beta_1+\beta_4
	\;\;\;\;\;\;\;\;\;
	\gamma_1+\gamma_3\ge\alpha_3+\alpha_4+\beta_1+\beta_3
	\\
	\gamma_1+\gamma_2\ge\alpha_3+\alpha_4+\beta_1+\beta_2
	\;\;\;\;\;\;\;\;\;
	\gamma_2+\gamma_4\ge\alpha_2+\alpha_4+\beta_3+\beta_4
	\\
	\gamma_2+\gamma_3\ge\alpha_2+\alpha_4+\beta_2+\beta_4
	\;\;\;\;\;\;\;\;\;
	\gamma_1+\gamma_4\ge\alpha_2+\alpha_4+\beta_2+\beta_4
	\\
	\gamma_1+\gamma_3\ge\alpha_2+\alpha_4+\beta_2+\beta_3
	\;\;\;\;\;\;\;\;\;
	\gamma_1+\gamma_3\ge\alpha_2+\alpha_4+\beta_1+\beta_4
	\\
	\gamma_1+\gamma_2\ge\alpha_2+\alpha_4+\beta_1+\beta_3
	\;\;\;\;\;\;\;\;\;
	\gamma_2+\gamma_3\ge\alpha_2+\alpha_3+\beta_3+\beta_4
	\\
	\gamma_1+\gamma_3\ge\alpha_2+\alpha_3+\beta_2+\beta_4
	\;\;\;\;\;\;\;\;\;
	\gamma_1+\gamma_2\ge\alpha_2+\alpha_3+\beta_2+\beta_3
	\\
	\gamma_1+\gamma_4\ge\alpha_1+\alpha_4+\beta_3+\beta_4
	\;\;\;\;\;\;\;\;\;
	\gamma_1+\gamma_3\ge\alpha_1+\alpha_4+\beta_2+\beta_4
	\\
	\gamma_1+\gamma_2\ge\alpha_1+\alpha_4+\beta_1+\beta_4
	\;\;\;\;\;\;\;\;\;
	\gamma_1+\gamma_3\ge\alpha_1+\alpha_3+\beta_3+\beta_4
	\\
	\gamma_1+\gamma_2\ge\alpha_1+\alpha_3+\beta_2+\beta_4
	\;\;\;\;\;\;\;\;\;
	\gamma_1+\gamma_2\ge\alpha_1+\alpha_2+\beta_3+\beta_4
\end{align*}
For $r = 3$:
\begin{align*}
	\gamma_2+\gamma_3+\gamma_4\ge\alpha_2+\alpha_3+\alpha_4+\beta_2+\beta_3+\beta_4
	\\
	\gamma_1+\gamma_3+\gamma_4\ge\alpha_2+\alpha_3+\alpha_4+\beta_1+\beta_3+\beta_4
	\\
	\gamma_1+\gamma_2+\gamma_4\ge\alpha_2+\alpha_3+\alpha_4+\beta_1+\beta_2+\beta_4
	\\
	\gamma_1+\gamma_2+\gamma_3\ge\alpha_2+\alpha_3+\alpha_4+\beta_1+\beta_2+\beta_3
	\\
	\gamma_1+\gamma_3+\gamma_4\ge\alpha_1+\alpha_3+\alpha_4+\beta_2+\beta_3+\beta_4
	\\
	\gamma_1+\gamma_2+\gamma_4\ge\alpha_1+\alpha_3+\alpha_4+\beta_1+\beta_3+\beta_4
	\\
	\gamma_1+\gamma_2+\gamma_3\ge\alpha_1+\alpha_3+\alpha_4+\beta_1+\beta_2+\beta_4
	\\
	\gamma_1+\gamma_2+\gamma_4\ge\alpha_1+\alpha_2+\alpha_4+\beta_2+\beta_3+\beta_4
	\\
	\gamma_1+\gamma_2+\gamma_3\ge\alpha_1+\alpha_2+\alpha_4+\beta_1+\beta_3+\beta_4
	\\
	\gamma_1+\gamma_2+\gamma_3\ge\alpha_1+\alpha_2+\alpha_3+\beta_2+\beta_3+\beta_4
\end{align*}

\subsection{Evaluation}
In the following figure \eqref{fig:fig4oursVsWintersWildes}, we compare our new bound in lemma \eqref{lem:ourlemma}, to Winter's bound \cite{winter2016tight} and the conjecture by Wilde \cite{wilde2020optimal}.
\begin{figure}[h]
	\centering
	\includegraphics[width=.9\linewidth]{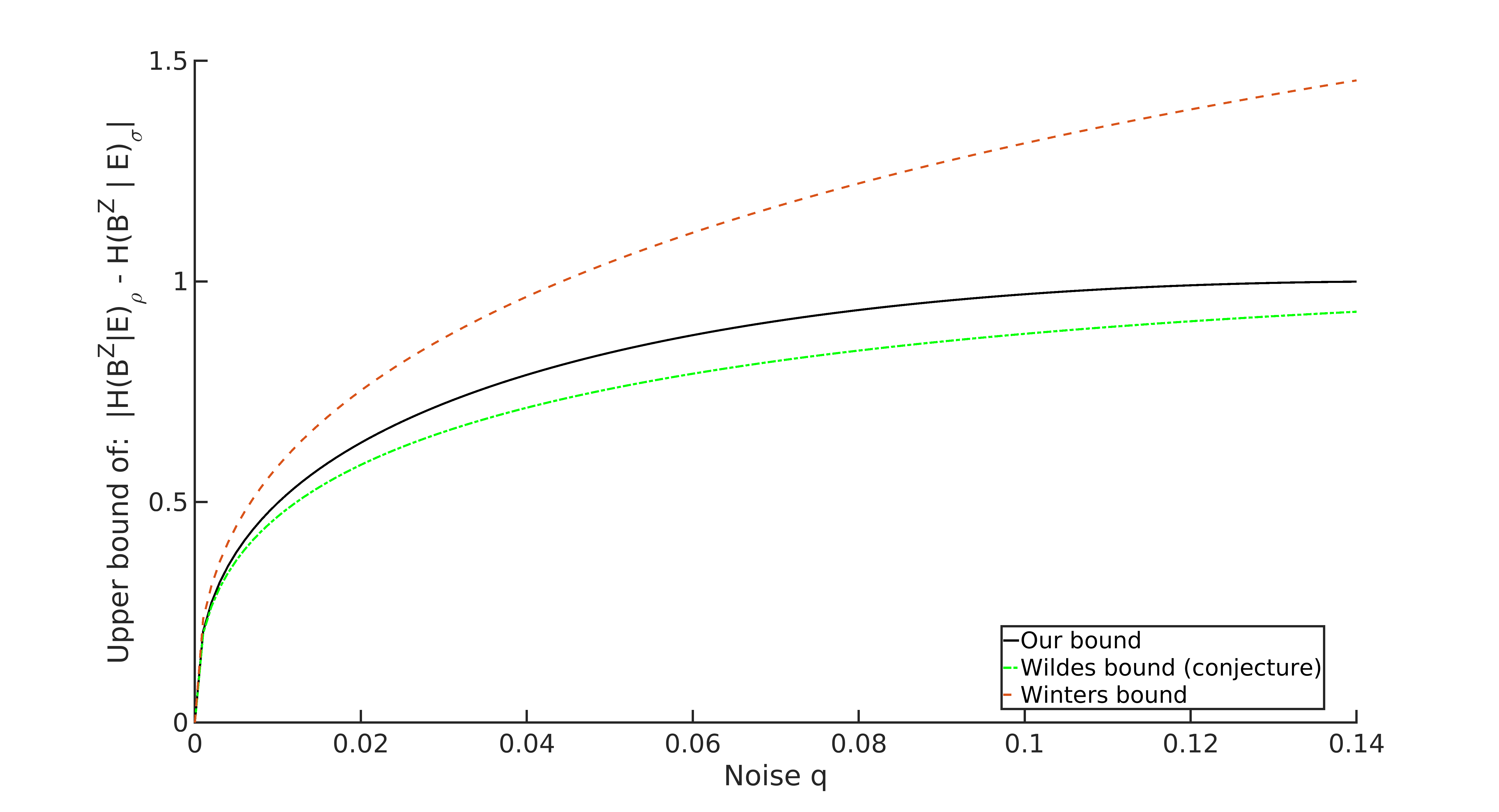}
	\caption{Comparison of our lemma vs Winter’s theorem \cite{winter2016tight} and Wilde’s
		conjecture \cite{wilde2020optimal}, for the upper bound of $|H(B^Z | E)_\rho - H(B^Z | E)_\sigma|$.  Our bound sits between the conjecture (bottom) and the proven result (top); since these are upper-bounds, lower is better. However, we note that our lemma is only for a certain class of state, whereas Winter's result (top), along with the conjectured bound (bottom), is applicable to any state.}
	\label{fig:fig4oursVsWintersWildes}
\end{figure}

Note that, the difficulty in using Horn's theorem for proving our lemma in \eqref{lem:ourlemma} for higher dimensions $(D > 2)$, with an appropriate modification on the right-hand side, lies in the size of the set of inequalities that Horn's theorem prescribes as the dimension increases. However, as pointed out in \cite{fulton2000eigenvalues}, as the dimension increases, the number of redundant inequalities in this set also increases rapidly, so there might be some hope of proving this lemma in higher dimensions too. We leave the discovery of such lemmas like the one we have presented here and their proofs involving Horn's theorem in arbitrary dimensions as interesting future work.

\end{document}